\documentclass[a4paper,12pt]{article}
\usepackage[left=3.17cm,right=3.17cm,top=2.54cm,
headheight=0.5cm,headsep=0.54cm,bottom=2.54cm,footskip=0.79cm
]{geometry}

\usepackage{amsmath,amsthm,amssymb,cite}

\usepackage{hyperref}

\newtheorem{corollary}{Corollary}
\newtheorem{definition}{Definition}
\newtheorem{lemma}{Lemma}
\newtheorem{theorem}{Theorem}

\begin{document}

\title{A note on unitary invariance of Connes spectral distances between quantum states}

\author{Bing-Sheng Lin$^{1,2,}$\thanks{Corresponding author. 
		E-mail address: sclbs@scut.edu.cn}, Ji-Hong Wang$^{3}$, Zhi-Kang You$^{1}$\\
	\small $^{1}$ School of Mathematics, South China University of Technology,
		Guangzhou 510641, China\\
	\small $^{2}$ Institute for Mathematics, Astrophysics and Particle Physics,\\
	\small  Radboud University Nijmegen, Heyendaalseweg 135, 6525 AJ Nijmegen, The Netherlands\\
	\small $^{3}$ Guangzhou University of Software, Guangzhou 510990, China
}

\date{November 9, 2026}

\maketitle

\begin{abstract}
In this paper, we study some properties of Connes spectral distances between quantum states under unitary transformations. We mainly focus on spectral triples with matrix algebras acting on finite-dimensional Hilbert spaces via unital representations. We find that there is a direct correspondence between the Lipschitz unit sphere and the set of optimal elements of finite Connes spectral distances.
Quantum trace distance is one of the most important unitarily invariant distance measures.
We prove that for any finite-dimensional matrix algebra, there exist finite-dimensional spectral triples for which the Connes spectral distances coincide with the quantum trace distances between quantum states.
We also construct finite-dimensional spectral triples in which the Lipschitz seminorm coincides with the operator norm for traceless Hermitian elements, and the Connes spectral distances between one-qubit states are exactly the quantum trace distances.
These results provide a better understanding of the unitary invariance of Connes spectral distance and the relationships between Connes spectral distance and other quantum distance measures. These concrete examples are valuable for studying the geometric structures of finite spectral triples and mathematical relations of qubits and other quantum states in the framework of noncommutative geometry.\\
  
  

\end{abstract}

\section{Introduction}
In the 1980s, Connes formulated the mathematically rigorous framework of noncommutative geometry \cite{Connes}.
The main objects in noncommutative geometry are spectral triples, which can be regarded as noncommutative counterparts of ordinary spaces.
Due to the noncommutativity, there is no traditional point in a spectral triple. This is different from ordinary commutative spaces.
In a spectral triple, a pure state is the
analog of a traditional point in a normal commutative space.
So there is no normal concept of distance between two points in a noncommutative space, but
we can calculate various distance measures between states, such as the
Connes distance \cite{Connes1}.
The Connes spectral distance between pure states corresponds to the geodesic distance
between points.
The Connes spectral distances in different types
of noncommutative spaces have already been studied in literature
\cite{Bimonte,Dai,Cagnache,Wallet,Martinetti1,Martinetti,DAndrea,Pythagoras,Pythagoras1,Franco,Scholtz,Chaoba,Revisiting,Kumar,Barrett,Chakraborty,Lin1,Lin2,Clare}.
For example, Cagnache \textit{et al.} computed
Connes spectral distances between the pure states that correspond to
eigenfunctions of the quantum harmonic oscillators on the Moyal
plane\cite{Cagnache}. Martinetti \textit{et al.} obtained the
spectral distance between coherent states in the so-called double Moyal plane
\cite{Martinetti}. D’Andrea \textit{et al.} have studied the Pythagoras’
theorem in noncommutative geometry \cite{Pythagoras}. They showed that for non-pure states
it is replaced by some Pythagoras inequalities. Franco and Wallet also studied
metrics and causality on Moyal planes \cite{Franco}. Scholtz and his collaborators have studied the Connes
spectral distances of harmonic oscillator states and also coherent states on the
Moyal plane and in fuzzy space \cite{Scholtz,Chaoba,Revisiting}.
Kumar \textit{et al.} used the Dirac eigen-spinor method to compute spectral distances on the doubled Moyal plane \cite{Kumar}. Chakraborty \textit{et al.} also studied the spectral
distance on the Lorentzian Moyal plane \cite{Chakraborty}.
Lin \textit{et al.} studied the relation between Connes spectral
distance and the noncommutativity of phase spaces \cite{Lin2}.

In quantum physics, most of the properties of physical systems are invariant
under unitary transformations. Connes spectral distances are also some essential
physical properties of quantum states. Therefore, it is interesting to study
some types of Connes spectral distances which are invariant under unitary transformations of the states.
Such unitary invariance is closely related to isometries of state spaces studied in other works, such as Ref.~\cite{Martinettip}.
The unitary invariance of spectral distances can be regarded as a natural property of a spectral triple.
Quantum trace distance is one of the most useful distance measures between
quantum states in quantum information science \cite{Nielsen}. It is also unitarily invariant.
So it is also very significant to study spectral triples whose Connes spectral distances coincide with quantum trace distances.

In the present work, we will study some properties and concrete examples of Connes spectral distances between quantum states under unitary transformations.
This paper is organized as follows. In Sec.~\ref{sec2}, we review the definitions of spectral triples and Connes spectral distances.
In Sec.~\ref{sec3}, we study some elementary properties of Connes spectral distances and optimal elements.
We find that a Hermitian element is an optimal element for the Connes spectral distance between some two distinct states if and only if the Lipschitz seminorm equals $1$.
In Sec.~\ref{sec4}, we study the unitary invariance of Connes spectral distances
of quantum states. In Sec.~\ref{sec5}, we prove that there are some finite spectral triples in which
the Connes spectral distances between states are exactly the quantum trace distances. In the case of $\mathbb{M}_2(\mathbb{C})$, the Lipschitz seminorm can be equal to the operator norm for traceless Hermitian elements.
In Sec.~\ref{sec6}, we explicitly construct some spectral triples whose Connes spectral distances coincide with quantum trace distances between one-qubit states. 
Sec.~\ref{sec7} presents the discussions and conclusions.

\section{Connes spectral distances in spectral triples}\label{sec2}

In the framework of noncommutative geometry \cite{Connes}, a noncommutative version of space corresponds to a spectral triple $(\mathcal{A}, \mathcal{H},
\mathcal{D})$, where $\mathcal{A}$ is an involutive algebra acting on a
Hilbert space $\mathcal{H}$ through a representation $\pi$, and the Dirac
operator $\mathcal{D}$ is a self-adjoint, densely defined operator on
$\mathcal{H}$ that satisfies the following conditions:

(1).\;\;$\mathcal{D}$ can be an unbounded operator but $[\mathcal{D},
\pi (a)]$ is bounded for every $a \in \mathcal{A}$;

(2).\;\;$\mathcal{D}$ has compact resolvent, for $\lambda \in
\mathbb{C} \backslash \mathbb{R}$, $(\mathcal{D}- \lambda)^{- 1}$ is compact
when the algebra $\mathcal{A}$ is unital or $\pi (a)  (\mathcal{D}-
\lambda)^{- 1}$ must be compact if the algebra is non-unital.

In the present work, we consider the finite-dimensional matrix algebra $\mathcal{A}=\mathbb{M}_n (\mathbb{C})$ with representations $\pi$ on
finite-dimensional Hilbert spaces $\mathcal{H}= \bigoplus_{i = 1}^m
\mathbb{C}^n$, $n \geqslant 2$, $m \geqslant 1$, and $\pi (a) \in
\mathbb{M}_{mn} (\mathbb{C}) = L (\mathcal{H})$.
In most cases, we also assume that the
representation $\pi$ is unital, namely, $\pi (\mathbb{I}_{n}) =\mathbb{I}_{mn}$. For example, $\pi (a) =\mathbb{I}_m \otimes
a$. So for any unitary element $U \in \mathcal{A}$, $\pi (U)$ is also unitary.

A state $\omega$ on the algebra $\mathcal{A}$ is a positive linear functional of norm one, i.e., $\omega (a^{\ast}
a) \geqslant 0$ for all $a \in \mathcal{A}$ and $\|\omega\|=1$.
A quantum
state $\omega$ can also be represented by a
density operator $\rho \in \mathcal{A}$, which is positive semidefinite,
self-adjoint, and has trace $1$. The action of the state $\omega$ on an element $e
\in \mathcal{A}$ can be written as \cite{Revisiting}
\begin{equation}
  \omega (e) = \mathrm{tr} (\rho e) .
\end{equation}
So the state $\omega$ can be represented by the corresponding density matrix
$\rho$. For convenience, we will denote the state $\omega$ by the
corresponding density matrix $\rho$ in the following content if there is no
confusion.

Suppose the quantum states $\omega_1$ and $\omega_2$ correspond to the density
matrices $\rho_1$ and $\rho_2$, respectively.

\begin{definition}
  \cite{Connes1} The Connes distance between quantum states $\omega_1$ and
  $\omega_2$ is
  \begin{equation}\label{csd}
    d (\omega_1, \omega_2) \equiv d (\rho_1, \rho_2) = \sup_{e \in B} |
    \mathrm{tr} (\rho_1 e) - \mathrm{tr} (\rho_2 e) | = \sup_{e \in B} |
    \mathrm{tr} (\Delta \rho \, e) |,
  \end{equation}
  where $\Delta \rho = \rho_1 - \rho_2$, and the set $B$ is defined by
  \begin{equation}
    B := \{e \in \mathcal{A}: \|[\mathcal{D}, \pi (e)]\|_{op} \leqslant 1\},
  \end{equation}
  where $\|A\|_{op}$ is the operator norm of $A$,
  \begin{equation}
    \|A\|_{op} \equiv \sup_{\psi \in \mathcal{H}, \| \psi \| = 1} \|A \psi \|,
    \qquad \|A\|^2 \equiv \mathrm{tr} (A^{\dag} A) = \langle A, A \rangle,
  \end{equation}
  and $\langle \cdot, \cdot \rangle \equiv \langle \cdot, \cdot
  \rangle_{\mathrm{HS}}$ is the Hilbert-Schmidt inner product, $\langle A, B
  \rangle \equiv \langle A, B \rangle_{\mathrm{HS}} := \mathrm{tr} (A^{\dag}
  B)$.
\end{definition}

The inequality $\|[\mathcal{D}, \pi (e)]\|_{op} \leqslant 1$ is the so-called
\textit{ball condition}. The map $L_{\mathcal{D}} (a) = \|[\mathcal{D}, \pi (a)]\|_{op}$ is
a Lipschitz seminorm.

Obviously, if $e \in B$, then we also have $- e \in B$, so we also have
\begin{equation}
  d (\rho_1, \rho_2) = \sup_{e \in B} [\mathrm{tr} (\rho_1 e) - \mathrm{tr}
  (\rho_2 e)] = \sup_{e \in B} [\mathrm{tr} (\Delta \rho \, e)] =
  \mathrm{tr} (\Delta \rho \, e_o),
\end{equation}
where $e_o$ is the so-called \textit{optimal element}.
Note that the last equality above makes sense only if the supremum is attained (for simplicity, we only consider these cases in the present work).
Since
Hermitian elements can give the supremum in the Connes spectral distance
functions \cite{Iochum}, 
one only needs to consider the optimal elements
$e_o$ being Hermitian.

Since $(\Delta \rho)^{\dag} = \Delta \rho$, we also have
\begin{equation}
  d (\rho_1, \rho_2) = \sup_{e \in B} (\mathrm{tr} [(\Delta \rho)^{\dag} 
  \, e]) = \sup_{e \in B} \langle \Delta \rho, e \rangle =
  \langle \Delta \rho, e_o \rangle .
\end{equation}
Obviously, there is $d (\rho_1, \rho_2) \geqslant 0$. In the present work, we
consider only cases in which $d (\rho_1, \rho_2) < + \infty$.

\section{Elementary properties of Connes spectral distance}\label{sec3}
In this section, we will study some basic but important properties of Connes spectral distances and the optimal elements. Some of these conclusions have already appeared in other research papers in a similar form. We have listed these simple properties here mainly for ease of reference in what follows.
\begin{lemma}\label{d0}\cite{Martinettif}
  For any states $\rho_1, \rho_2$, the Connes spectral distance $d (\rho_1,
  \rho_2) = 0$ if and only if $\rho_1 = \rho_2$. On the other hand, for any
  different states $\rho_1, \rho_2$, we necessarily have $d (\rho_1, \rho_2) > 0$.
\end{lemma}

\begin{proof}
Obviously, for the same states $\rho_1 = \rho_2$,
there is $d (\rho_1, \rho_2) = 0$.

For different states $\rho_1, \rho_2$, there is $\Delta \rho \neq 0$, and $\|
\Delta \rho \| > 0$. If $\|[\mathcal{D}, \pi (\Delta \rho)]\|_{op} = 0$, then
for any number $c > 0$, we always have $\|[\mathcal{D}, \pi (c \Delta
\rho)]\|_{op} = c \|[\mathcal{D}, \pi (\Delta \rho)]\|_{op} = 0$. Therefore,
there is $e_c \equiv c \Delta \rho \in B$ for $\forall c > 0$. But in this
case, we have
\begin{equation}
  d (\rho_1, \rho_2) = \sup_{e \in B} [\mathrm{tr} (\Delta \rho
  \, e)] \geqslant \mathrm{tr} (\Delta \rho \, e_c)
  = c \| \Delta \rho \|^2, \quad \forall c > 0.
\end{equation}
So there should be $d (\rho_1, \rho_2) = + \infty$.

If $\|[\mathcal{D}, \pi (\Delta \rho)]\|_{op} \neq 0$, then one can choose
\begin{equation}
  e = \frac{\Delta \rho}{\|[\mathcal{D}, \pi (\Delta \rho)]\|_{op}},
\end{equation}
and there is $\|[\mathcal{D}, \pi (e)]\|_{op} = 1$. So we have $e \in B$, and
\begin{equation}
  | \mathrm{tr} (\Delta \rho \, e) | = \frac{| \mathrm{tr}
  (\Delta \rho \Delta \rho) |}{\|[\mathcal{D}, \pi (\Delta \rho)]\|_{op}} = \frac{\| \Delta \rho \|^2}{\|[\mathcal{D}, \pi (\Delta
  \rho)]\|_{op}} > 0.
\end{equation}
Therefore, for different states $\rho_1, \rho_2$, there should be some $e \in
B$ which satisfies $| \mathrm{tr} (\Delta \rho \, e) | > 0$. So
we have
\begin{equation}
  d (\rho_1, \rho_2) = \sup_{e \in B} | \mathrm{tr} (\Delta \rho
  \, e) | > 0.
\end{equation}
\end{proof}

\begin{lemma}
  The spectral triples $(\mathcal{A}, \mathcal{H}, \mathcal{D})$ and
  $(\mathcal{A}, \mathcal{H}, \lambda \mathbb{I}+\mathcal{D})$ have the same
  metric, where $\lambda$ is any real number.
\end{lemma}

\begin{proof}
Obviously, there is
\begin{equation}
  B = \{e \in \mathcal{A}: \|[\mathcal{D}, \pi (e)]\|_{op} \leqslant 1\} = \{e
  \in \mathcal{A}: \|[\lambda \mathbb{I}+\mathcal{D}, \pi (e)]\|_{op}
  \leqslant 1\} .
\end{equation}
So the spectral triples $(\mathcal{A}, \mathcal{H}, \mathcal{D})$ and
$(\mathcal{A}, \mathcal{H}, \lambda \mathbb{I}+\mathcal{D})$ have the same
ball conditions and metric.
\end{proof}

So in most cases, one can simply take the Dirac
operator $\mathcal{D}$ to be Hermitian and traceless,
\begin{equation}
  \mathrm{tr} (\mathcal{D}) = 0.
\end{equation}

Furthermore, if we assume that the representation $\pi$ is
\textit{unital}, namely, $\pi (\mathbb{I}) =\mathbb{I}$, we also
have the following result.

\begin{lemma}
  \label{tloe}For any states $\rho_1, \rho_2$, if $e_o$ is the corresponding
  optimal element, then $\lambda \mathbb{I}+ e_o$ is also an optimal element,
  where $\lambda$ is any real number.
\end{lemma}

\begin{proof}
For the optimal element $e_o$, there is
\begin{equation}
  \pi (\lambda \mathbb{I}+ e_o) = \lambda \pi (\mathbb{I}) + \pi (e_o) =
  \lambda \mathbb{I}+ \pi (e_o),
\end{equation}
and
\begin{equation}
  \|[\mathcal{D}, \pi (\lambda \mathbb{I}+ e_o)]\|_{op} = \|[\mathcal{D},
  \lambda \mathbb{I}+ \pi (e_o)]\|_{op} = \|[\mathcal{D}, \pi (e_o)]\|_{op}
  \leqslant 1.
\end{equation}
So $\lambda \mathbb{I}+ e_o \in B$. Since $\mathrm{tr} (\Delta\rho)=\mathrm{tr} (\rho_1-\rho_2)=0$, we also have
\begin{equation}
  \mathrm{tr} [\Delta \rho (\lambda \mathbb{I}+ e_o)] = \mathrm{tr} (\Delta
  \rho \cdot \lambda \mathbb{I}) + \mathrm{tr} (\Delta \rho e_o) = \mathrm{tr}
  (\Delta \rho e_o) = \sup_{e \in B} [\mathrm{tr} (\Delta \rho e)] .
\end{equation}
Therefore, $\lambda \mathbb{I}+ e_o$ is also an optimal element for the states
$\rho_1, \rho_2$.
\end{proof}

So if the representation $\pi$ is unital, one can simply consider the optimal elements $e_o \in B$
being Hermitian and traceless,
\begin{equation}
  \mathrm{tr} (e_o) = 0.
\end{equation}
\begin{lemma}
  Let $d_{\mathcal{D}}
  (\rho_1, \rho_2)$ denote the spectral distance between the states $\rho_1, \rho_2$ in the spectral
  triple $(\mathcal{A}, \mathcal{H}, \mathcal{D})$. Then there is
  \begin{equation}
    d_{\lambda \mathcal{D}} (\rho_1, \rho_2) = \frac{1}{| \lambda |}
    d_{\mathcal{D}} (\rho_1, \rho_2),
  \end{equation}
  where $\lambda$ is any nonzero real number.
\end{lemma}

\begin{proof}
For any nonzero real number $\lambda$, there are
\begin{equation}
  \|[\lambda \mathcal{D}, \pi (e)]\|_{op} = \|[\mathcal{D}, \lambda \pi
  (e)]\|_{op} = \|[\mathcal{D}, \pi (\lambda e)]\|_{op} .
\end{equation}
So if the element $e \in B$ in the spectral triple $(\mathcal{A}, \mathcal{H},
\mathcal{D})$, namely, $\|[\mathcal{D}, \pi (e)]\|_{op} \leqslant 1$, then we
have $\|[\lambda \mathcal{D}, \pi \left( \frac{1}{\lambda} e \right)]\|_{op}
\leqslant 1$, namely, $e' = \frac{1}{\lambda} e \in B'$ in the spectral triple
$(\mathcal{A}, \mathcal{H}, \lambda \mathcal{D})$; vice versa.

So for any states $\rho_1, \rho_2$, there is
\begin{eqnarray}
  d_{\lambda \mathcal{D}} (\rho_1, \rho_2) & = & \sup_{e' \in B'} |
  \mathrm{tr} (\Delta \rho \, e') | = \sup_{e \in B} | \mathrm{tr} (\Delta \rho \cdot \frac{1}{\lambda} e)
  | \nonumber\\
  & = & \frac{1}{| \lambda |} \sup_{e \in B} | \mathrm{tr} (\Delta \rho
  \, e) | = \frac{1}{| \lambda |} d_{\mathcal{D}} (\rho_1, \rho_2) . 
\end{eqnarray}
\end{proof}

\begin{lemma}\cite{Martinettip}
  Let $\rho_1$ and $\rho_2$ be two distinct states with finite spectral distance, and let $e_o$ be a corresponding optimal element. Then $\|[\mathcal{D}, \pi (e_o)]\|_{op} > 0$.
\end{lemma}

\begin{proof}
Since $\rho_1, \rho_2$ are different states, from Lemma \ref{d0}, there is
$d (\rho_1, \rho_2) > 0$. For the optimal element $e_o$, if there is
$\|[\mathcal{D}, \pi (e_o)]\|_{op} = 0$, and
\begin{equation}
  0 < d (\rho_1, \rho_2) = \sup_{e \in B} [\mathrm{tr} (\Delta \rho
  \, e)] = \mathrm{tr} (\Delta \rho e_o) < + \infty,
\end{equation}
then there should be some real number $t > 1$,
\begin{equation}
  \|[\mathcal{D}, \pi (te_o)]\|_{op} = t \|[\mathcal{D}, \pi (e_o)]\|_{op} = 0
  \leqslant 1.
\end{equation}
So there is $te_o \in B$. But
\begin{equation}
  \mathrm{tr} (\Delta \rho te_o) = t \cdot \mathrm{tr} (\Delta \rho e_o) >
  \mathrm{tr} (\Delta \rho e_o) = \sup_{e \in B} [\mathrm{tr} (\Delta \rho
  \, e)],
\end{equation}
which is impossible. Therefore, if the spectral distance is finite, there should be
$\|[\mathcal{D}, \pi (e_o)]\|_{op} \neq 0$. So we have
$\|[\mathcal{D}, \pi (e_o)]\|_{op} > 0$ for two different states $\rho_1, \rho_2$.
\end{proof}

Obviously, if $\pi (\mathcal{A}) \subseteq \mathcal{C} (\mathcal{D})$, where
  $\mathcal{C} (\mathcal{D})$ is the centralizer of $\mathcal{D}$ in $L
  (\mathcal{H})$, then the Connes spectral distances between any states $\rho_1, \rho_2$ are $d (\rho_1, \rho_2) =
  0$ or $+ \infty$.
The trivial example is $\mathcal{D}= c\mathbb{I}$. In the present work, 
we consider only spectral triples for which all Connes spectral distances are finite.

\begin{theorem}\label{thm1}\cite{Iochum}
  If the spectral distance of two different states $\rho_1, \rho_2$ is finite,
  then for the corresponding optimal elements $e_o$, there is
  \begin{equation}
    \|[\mathcal{D}, \pi (e_o)]\|_{op} = 1.
  \end{equation}
\end{theorem}

\begin{proof}
Since the spectral distance is finite, for the
optimal element $e_o$, there should be $\|[\mathcal{D}, \pi (e_o)]\|_{op} >
0$. If
\begin{equation}
  0 < \|[\mathcal{D}, \pi (e_o)]\|_{op} = t < 1,
\end{equation}
then we also have
\begin{equation}
  \left\| \left[ \mathcal{D}, \pi \left( \frac{e_o}{t} \right) \right]
  \right\|_{op} = \|[\mathcal{D}, \pi (e_o')]\|_{op} = \frac{1}{t}
  \|[\mathcal{D}, \pi (e_o)]\|_{op} = 1.
\end{equation}
So there is $e_o' = \frac{e_o}{t} \in B$, but
\begin{eqnarray}
  \mathrm{tr} (\Delta \rho e_o') & = & \mathrm{tr} (\Delta \rho \frac{e_o}{t})
  = \frac{1}{t} \mathrm{tr} (\Delta \rho e_o)  >  \mathrm{tr} (\Delta \rho e_o) \nonumber\\
  & = & \sup_{e \in B} [\mathrm{tr} (\Delta \rho \, e)]
  = d (\rho_1, \rho_2) . 
\end{eqnarray}
This is a contradiction. So for any optimal element $e_o$ of different states,
there should be $\|[\mathcal{D}, \pi (e_o)]\|_{op} =1$.
\end{proof}

Since for a given state $\rho$, any element $e_o$ can be considered an optimal element,
\begin{equation}
  d (\rho, \rho) = \mathrm{tr} (\Delta \rho e_o) = 0.
\end{equation}
Therefore, one may choose an optimal element $e_o$ which
satisfies the relation $\|[\mathcal{D}, \pi (e_o)]\|_{op} = 1$.

\begin{definition}
  Let $(\mathcal{A}, \mathcal{H}, \mathcal{D})$ be a spectral triple with
  finite spectral distances, and denote the set of all optimal elements by
  $B_o$. Obviously, there is
  \begin{equation}
    \|[\mathcal{D}, \pi (e)]\|_{op} = 1, \quad \forall e \in B_o .
  \end{equation}
\end{definition}

We also define the set
\begin{equation}
  B_1 = \{e \in \mathcal{A}: \|[\mathcal{D}, \pi (e)]\|_{op} = 1\} .
\end{equation}
Obviously, there is
\begin{equation}\label{bbo}
  B_o \subseteq B_1 .
\end{equation}
In the present work, we consider only spectral triples with finite spectral
distances. So in most of the cases, one can only consider the subset $B_o$ or
$B_1$ instead of $B$.

Conversely, we also have the following result.
\begin{theorem}
	\label{thm2}Let $(\mathcal{A}, \mathcal{H}, D)$ be a
	finite-dimensional spectral triple with $\mathcal{A}=\mathbb{M}_n
	(\mathbb{C})$, $n \ge 2$, and $\pi$ is a unital $\ast$-representation. Then
	every Hermitian element $e \in \mathcal{A}$ satisfying
	\begin{equation}
		\|[D, \pi (e)]\|_{o p} = 1 \label{eq:unit-L}
	\end{equation}
	is an optimal element for the Connes spectral distance between some two
	distinct states.
\end{theorem}

\begin{proof}
	See Appendix \ref{thm2a} for a detailed proof. 
\end{proof}

By combining the results of Theorems \ref{thm1} and \ref{thm2}, one can directly obtain the following necessary and sufficient condition for a Hermitian element to be optimal.
\begin{corollary}
	Let $(\mathcal{A}, \mathcal{H}, D)$ be a
	finite-dimensional spectral triple with $\mathcal{A}=\mathbb{M}_n
	(\mathbb{C})$, $n \ge 2$, and $\pi$ is a unital $\ast$-representation. Assume that the Connes spectral distances are all finite. Then
	a Hermitian element $e \in \mathcal{A}$ satisfies
	\begin{equation}
		\|[D, \pi (e)]\|_{o p} = 1
	\end{equation}
	if and only if it is an optimal element for the Connes spectral distance between some two 	distinct states.
\end{corollary}

This establishes a direct correspondence between the Lipschitz unit sphere and
the set of optimal elements of Connes spectral distances.

\begin{lemma}\cite{Martinettip}
  Let $\mathcal{C} (\mathcal{D}) \subset L (\mathcal{H})$ be the centralizer
  of $\mathcal{D}$. If all Connes spectral distances in the spectral triples
  $(\mathcal{A}, \mathcal{H}, \mathcal{D})$ are finite, then the subset $\pi
  (\mathcal{A}) \cap \mathcal{C} (\mathcal{D})$ is trivial, namely, $\pi
  (\mathcal{A}) \cap \mathcal{C} (\mathcal{D}) = \{\lambda \mathbb{I}\}$.
\end{lemma}

\begin{proof}
If the subset $\pi (\mathcal{A}) \cap \mathcal{C}
(\mathcal{D})$ is non-trivial, namely, 
there exists an element $e \in
\mathcal{A}$ that is not a scalar multiple of $\mathbb{I}$
and $\pi (e) \in \mathcal{C}
(\mathcal{D})$. So we have $\|[\mathcal{D}, \pi (e)]\|_{op} = 0$, and $e \in
B$. By the separating property of states, there do exist some states $\rho_1,
\rho_2$ with $\mathrm{tr} (\rho_1  \, e) \neq \mathrm{tr} (\rho_2
\, e)$. Therefore, we have
\begin{equation}
  | \mathrm{tr} (\rho_1  \, e) - \mathrm{tr} (\rho_2 
  \, e) | = | \mathrm{tr} (\Delta \rho \, e) | =
  \delta > 0.
\end{equation}
Obviously, for any real number $k > 0$, we also have $\|[\mathcal{D}, \pi
(ke)]\|_{op} = 0$. So $e' = ke \in B$, but
\begin{equation}
  d (\rho_1, \rho_2) = \sup_{e\in B} | \mathrm{tr} (\Delta \rho e) |
  \geqslant | \mathrm{tr} (\Delta \rho ke) | = k \delta .
\end{equation}
So there should be $d (\rho_1, \rho_2) = + \infty$.

Therefore, if the spectral distances are finite, then the subset $\pi
(\mathcal{A}) \cap \mathcal{C} (\mathcal{D})$ is
trivial.
\end{proof}

\begin{corollary}
  If $\mathcal{A}=\mathbb{M}_n (\mathbb{C})$, $\mathcal{H}=\mathbb{C}^n$, $n
  \geqslant 2$, $\mathcal{D} \in \mathbb{M}_n (\mathbb{C})$, and the
  representation $\pi (a) = a \in \mathbb{M}_n (\mathbb{C})$, then the Connes
  spectral distances in the spectral triple $(\mathcal{A}, \mathcal{H},
  \mathcal{D})$ cannot all be finite.
\end{corollary}

\begin{proof}
For any matrix $\mathcal{D}$, the centralizer
$\mathcal{C} (\mathcal{D}) \subseteq \pi (\mathcal{A}) =\mathbb{M}_n
(\mathbb{C})$ is non-trivial. So the subset $\pi (\mathcal{A}) \cap
\mathcal{C} (\mathcal{D})$ is non-trivial, and then Connes spectral distances
in the spectral triple $(\mathcal{A}, \mathcal{H}, \mathcal{D})$ cannot all
be finite.
\end{proof}

A similar conclusion can also be found in Ref.~\cite{Martinettip}.

For example, consider the spectral triple $(\mathbb{M}_2 (\mathbb{C}),
\mathbb{C}^2, \mathcal{D})$ with the Dirac operator
\begin{equation}
  \mathcal{D} = \left(\begin{array}{cc}
    0 & 1\\
    1 & 0
  \end{array}\right),
\end{equation}
and the representation $\pi (a) = a \in \mathbb{M}_2 (\mathbb{C})$. For the element
\begin{equation}
  e_1 = \left(\begin{array}{cc}
    0 & 1\\
    1 & 0
  \end{array}\right) \in \mathbb{M}_2 (\mathbb{C}),
\end{equation}
there is $\|[\mathcal{D}, \pi (e_1)]\|_{op} = 0$, and $k e_1 \in B$ for
any real number $k > 0$.

Consider two states
\begin{equation}
  \rho_1 = |1 \rangle \langle 1| = \left(\begin{array}{cc}
    1 & 0\\
    0 & 0
  \end{array}\right), \quad \rho_2 = | + \rangle \langle + | = \frac{1}{2}
  \left(\begin{array}{cc}
    1 & 1\\
    1 & 1
  \end{array}\right),
\end{equation}
and
\begin{equation}
  \quad \Delta \rho = \rho_1 - \rho_2 = \frac{1}{2} \left(\begin{array}{cc}
    1 & - 1\\
    - 1 & - 1
  \end{array}\right).
\end{equation}
There is
\begin{equation}
  d (\rho_1, \rho_2) = \sup_{e \in B} | \mathrm{tr} (\Delta \rho e) |
  \geqslant | \mathrm{tr} (\Delta \rho k e_1) | = k | \mathrm{tr} (\Delta \rho
  e_1) | = k.
\end{equation}
So we have $d (\rho_1, \rho_2) = + \infty$.

\section{Unitary invariance of Connes spectral distances between quantum states}\label{sec4}

In quantum physics, most of the properties of physical systems are invariant
under unitary transformations. Connes spectral distances are also essential
physical properties of quantum states. Therefore, it is interesting to study
some types of Connes spectral distances which are unitarily invariant. The
unitary invariance of spectral distances can be regarded as a natural property
of a spectral triple. This is also a special kind of isometry of the state space being studied in some other literature, such as Ref.~\cite{Martinettip}.

\begin{definition}
  Assume that the spectral distances in the spectral triple $(\mathcal{A},
  \mathcal{H}, \mathcal{D})$ are all finite. For any states $\rho_1, \rho_2
  \in \mathcal{A}$ and any unitary operator $U \in \mathcal{U} (\mathcal{A})$,
  denote $\rho_1' = U \rho_1 U^{\dag}$, $\rho_2' = U \rho_2 U^{\dag}$. If the
  spectral distances satisfy
  \begin{equation}
    d (\rho_1', \rho_2') = d (\rho_1, \rho_2),
  \end{equation}
  we say that the spectral distances in the spectral triple $(\mathcal{A},
  \mathcal{H}, \mathcal{D})$ are \textit{unitarily invariant}.
\end{definition}

Obviously, if $\mathcal{A}$ is commutative, we have $\rho_i' = \rho_i$, and
the spectral distances are unitarily invariant.

First, let us consider the sets $B$ and $B_o, B_1$ under unitary transformations.
In general, the set $B$ is not necessarily closed under unitary conjugation. That is, $e \in B$ does not necessarily imply $UeU^{\dag} \in B$ for every unitary $U \in
\mathcal{U} (\mathcal{A})$. Similarly, the sets $B_o,
B_1$ are also not necessarily closed under unitary conjugations.

\begin{lemma}
	Let $(\mathcal{A}, \mathcal{H}, \mathcal{D})$ be a spectral triple with finite Connes spectral distances. If the set $B_o$ of optimal elements is closed under unitary conjugation, then the Connes spectral distances of the states are unitarily invariant.
\end{lemma}

\begin{proof}
If the optimal element set $B_o$ is closed under
unitary conjugations, namely, $e_o \in B_o$ if and only if $Ue_o U^{\dag} \in
B_o$, where $U \in \mathcal{U} (\mathcal{A})$ is any unitary operator, then
for any two states $\rho_1, \rho_2$, let $\rho_i' = U \rho_i U^{\dag}$, $\Delta
\rho' = U \Delta \rho U^{\dag}$. Then we have
\begin{eqnarray}
  d (\rho_1, \rho_2) & = & \sup_{e \in B_o} | \mathrm{tr} (\Delta \rho e) |
  = \sup_{UeU^{\dag} \in B_o} | \mathrm{tr} (\Delta \rho e) | = \sup_{UeU^{\dag} \in B_o} | \mathrm{tr} (U \Delta \rho U^{\dag} \cdot
  UeU^{\dag}) | \nonumber\\
  & = & \sup_{e' \in B_o} | \mathrm{tr} (\Delta \rho' e') | = d (\rho_1', \rho_2') . 
\end{eqnarray}
Thus, the Connes spectral distances are unitarily invariant.
\end{proof}

Similarly, one can prove the following result.

\begin{lemma}
  If the set $B_1$ (or $B$) is closed under unitary conjugations, then the
  Connes spectral distances are unitarily invariant.
\end{lemma}

\begin{corollary}
  If the set $B_1$ (or $B$) is closed under unitary conjugations, so is the
  optimal element set $B_o$.
\end{corollary}

\begin{proof}Since the set $B_1$ is closed under unitary
conjugation, the Connes spectral distances are unitarily invariant. For any
optimal element $e_o \in B_o$ and any unitary operator $U$,
there exist states $\rho_1, \rho_2$ for which $e_o$ is an optimal element, and $\rho_i' = U \rho_i
U^{\dag}$,
\begin{eqnarray}
  \mathrm{tr} (\Delta \rho e_o) & = & d (\rho_1, \rho_2) = \mathrm{tr} (U \Delta \rho U^{\dag} \cdot Ue_o U^{\dag}) = \mathrm{tr} (\Delta \rho' e_o') . \nonumber\\
  & = & d (\rho_1', \rho_2') = \sup_{e \in B_1} [\mathrm{tr} (\Delta \rho' e)] . 
\end{eqnarray}
Since $e_o \in B_o \subseteq B_1$, we have $Ue_o U^{\dag} = e_o' \in B_1$. So
$Ue_o U^{\dag} = e_o'$ is the corresponding optimal element of $\rho_1',
\rho_2'$. This means that $Ue_o U^{\dag} = e_o' \in B_o$. Therefore, the
optimal element set $B_o$ is closed under unitary
conjugations.
\end{proof}

\begin{theorem}\label{lu}
  If the operator norms $\|[\mathcal{D}, \pi (e)]\|_{op}$ is invariant under unitary conjugation, that is, if $\|[\mathcal{D}, \pi
  (e)]\|_{op} = \|[\mathcal{D}, \pi (UeU^{\dag})]\|_{op}$ for every $e\in\mathcal{A}$ and every unitary $U \in \mathcal{U} (\mathcal{A})$,
   then the Connes
  spectral distances are unitarily invariant.
\end{theorem}

\begin{proof}If $\|[\mathcal{D}, \pi (e)]\|_{op}$ is invariant
under any unitary transformations of the elements $e$, then the sets $B$,
$B_1$ and $B_o$ are all closed under unitary conjugations. Therefore, the
Connes spectral distances are unitarily invariant.
\end{proof}

This also means that the Lipschitz seminorm $L (a) = \|[\mathcal{D}, \pi
(a)]\|_{op}$ is invariant under unitary conjugation.
One simple case is $[\mathcal{D}, \pi (U)]=0$, and this is also studied in Ref.~\cite{Martinettip}.

Furthermore, if the representation $\pi$ is unital, namely,
$\pi (\mathbb{I}) =\mathbb{I}$, we have $\pi (UU^{\dag}) = \pi (\mathbb{I}) =
\pi (U) \pi (U^{\dag}) = \pi (U) \pi (U)^{\dag} =\mathbb{I}$. In this case,
$\pi (U)$ is also unitary. We have the following result.

\begin{corollary}
  If for any unitary operator $U$, there is $\|[\pi (U)\mathcal{D} \pi
  (U^{\dag}), \pi (e)]\|_{op} = \|[\mathcal{D}, \pi (e)]\|_{op}$, then the
  Connes spectral distances are unitarily invariant.
\end{corollary}

\begin{proof}
For any unitary operator $U$, we have
\begin{eqnarray}
  \|[\mathcal{D}, \pi (e)]\|_{op} & = & \|[\pi (U)\mathcal{D} \pi (U^{\dag}),
  \pi (e)]\|_{op} = \| \pi (U^{\dag}) [\pi (U)\mathcal{D} \pi (U^{\dag}), \pi (e)] \pi
  (U)\|_{op} \nonumber\\
  & = & \|[\mathcal{D}, \pi (U^{\dag}) \pi (e) \pi (U)]\|_{op} = \|[\mathcal{D}, \pi (U^{\dag} eU)]\|_{op} . 
\end{eqnarray}
Thus, $\|[\mathcal{D}, \pi (e)]\|_{op}$ is invariant under unitary conjugation of $e$. By Theorem \ref{lu}, the Connes spectral distances are unitarily invariant.
\end{proof}

\begin{corollary}
  If for any unitary operators $U, U' \in \mathcal{U} (\mathcal{A})$, there is
  $\|[\mathcal{D}, \pi (e)]\|_{op} = \|[\pi (U)\mathcal{D} \pi (U^{\dag}), \pi
  (U' {eU'}^{\dag})]\|_{op}$, then the Connes spectral distances are unitarily
  invariant.
\end{corollary}

\begin{theorem}
  \label{udu}The Connes spectral distances in the spectral triple
  $(\mathcal{A}, \mathcal{H}, \mathcal{D})$ are unitarily invariant if and only
  if the spectral triples $(\mathcal{A}, \mathcal{H}, \mathcal{D})$ and
  $(\mathcal{A}, \mathcal{H}, \pi (U)\mathcal{D} \pi (U^{\dag}))$ have the
  same metric for any unitary operator $U \in \mathcal{U} (\mathcal{A})$.
\end{theorem}

\begin{proof}
Let
$\mathcal{D}' = \pi (U) \mathcal{D} \pi (U^{\dag})$ and define $B'$ by $B' \equiv \{e \in \mathcal{A}:\|[\mathcal{D}', \pi (e)]\|_{op}
\leqslant 1\}$.
Obviously, for the given unitary operator $U$, if $e \in B$ in the
spectral triple $(\mathcal{A}, \mathcal{H}, \mathcal{D})$, then there is $e' =
UeU^{\dag} \in B'$ in $(\mathcal{A}, \mathcal{H}, \mathcal{D}')$; vice
versa.

If the Connes spectral distances in the spectral
triple $(\mathcal{A}, \mathcal{H}, \mathcal{D})$ are unitarily invariant, then
for any states $\rho_1, \rho_2$, and $\rho_i' = U^{\dag} \rho_i U$, there is
\begin{eqnarray}
  d_{\mathcal{D}} (\rho_1, \rho_2) & = & \sup_{e \in B} | \mathrm{tr} (\Delta
  \rho e) | = d_{\mathcal{D}} (\rho_1', \rho_2') = \sup_{e \in B} | \mathrm{tr} (\Delta \rho' e) | \nonumber\\
  & = & \sup_{e \in B} |
  \mathrm{tr} (U^{\dag} \Delta \rho U e) | = \sup_{e \in B} | \mathrm{tr} (\Delta \rho U eU^{\dag}) | = \sup_{e
  \in B} | \mathrm{tr} (\Delta \rho e') | \nonumber\\
  & = & \sup_{e' \in B'} | \mathrm{tr} (\Delta \rho e') | = d_{\mathcal{D}'} (\rho_1, \rho_2) . 
\end{eqnarray}
Hence, any two states $\rho_1, \rho_2\in\mathcal{A}$ have the same spectral distances in
the spectral triples $(\mathcal{A}, \mathcal{H}, \mathcal{D})$ and
$(\mathcal{A}, \mathcal{H}, \pi (U)\mathcal{D} \pi (U^{\dag}))$. So
$(\mathcal{A}, \mathcal{H}, \mathcal{D})$ and $(\mathcal{A}, \mathcal{H}, \pi
(U)\mathcal{D} \pi (U^{\dag}))$ have the same metric.

On the other hand, for any two states $\rho_1, \rho_2\in\mathcal{A}$, if they have the same
spectral distances in these two spectral triples, then we have
\begin{eqnarray}
  d_{\mathcal{D}} (\rho_1, \rho_2) & = & \sup_{e \in B} | \mathrm{tr} (\Delta
  \rho \, e) | = d_{\mathcal{D}'} (\rho_1, \rho_2) = \sup_{e' \in B'} | \mathrm{tr} (\Delta \rho e') | \nonumber\\
  & = & \sup_{e \in B} | \mathrm{tr} (\Delta \rho UeU^{\dag}) | = \sup_{e \in B} | \mathrm{tr} (U^{\dag} \Delta \rho U \,
  e) | = d_{\mathcal{D}} (\rho_1', \rho_2') . 
\end{eqnarray}
Therefore, the Connes spectral distances between any states $\rho_1, \rho_2$
are unitarily invariant.
\end{proof}

\begin{corollary}
  The Connes spectral distances in the spectral triple $(\mathcal{A},
  \mathcal{H}, \mathcal{D})$ are unitarily invariant if and only if the Connes
  spectral distances in $(\mathcal{A}, \mathcal{H}, \pi (U)\mathcal{D} \pi
  (U^{\dag}))$ are unitarily invariant, where $U \in \mathcal{U} (\mathcal{A})$
  is any unitary operator.
\end{corollary}

Some other articles have also investigated isometries of certain classes of spectral triples, such as Ref.~\cite{Bassi}.

\section{Spectral triples with quantum trace distances}\label{sec5}

Quantum trace distance is one of the most useful distance measures between
quantum states in quantum information science. It is also unitarily invariant.
So it is also very significant to study spectral triples with Connes spectral
distances which equal quantum trace distances.

\begin{definition}
  \cite{Nielsen} For two states $\rho_1, \rho_2$, the quantum trace distance
  is defined as
  \begin{equation}
    d_{\mathrm{tr}} (\rho_1, \rho_2) = \| \rho_1 - \rho_2 \|_1 = \| \Delta
    \rho \|_1 = \mathrm{tr} \sqrt{\Delta \rho^{\dag} \Delta \rho},
  \end{equation}
  where the trace norm $\|A\|_1 := \mathrm{tr} \sqrt{A^{\dag} A}$. Note that we have omitted
  the coefficient $\frac{1}{2}$ for simplicity.
\end{definition}

\begin{theorem}\label{qtd}
  \cite{Nielsen} For two states $\rho_1, \rho_2$, the quantum trace distance
  is equal to
  \begin{equation}\label{dtr}
    d_{\mathrm{tr}} (\rho_1, \rho_2) = \sup_{-\mathbb{I} \leqslant P \leqslant
    \mathbb{I}} \left| \mathrm{tr} [P (\rho_1 - \rho_2)] \right| = \sup_{\|P\|_{op}
    \leqslant 1} | \mathrm{tr} (P \Delta \rho) |,
  \end{equation}
  where $P$ ranges over Hermitian operators.
\end{theorem}

For any Hermitian element $e \in \mathcal{A}$, let $\lambda_1, \ldots,
\lambda_n$ be its eigenvalues, and define a new Hermitian operator by
\begin{equation}
	e_{\ast} = e - \frac{\lambda_{\max} (e) + \lambda_{\min} (e)}{2}
	\mathbb{I}_n .
\end{equation}
Thus, the eigenvalues of $e_{\ast}$ lie in the interval $\left[ -
\frac{\lambda_{\max} (e) - \lambda_{\min} (e)}{2}, \hspace{0.17em}
\frac{\lambda_{\max} (e) - \lambda_{\min} (e)}{2} \right]$, and there is
\begin{equation}
	\|e_{\ast} \|_{op} = \frac{\lambda_{\max} (e) - \lambda_{\min}
		(e)}{2}.
	\label{eq:b-contraction}
\end{equation}

\begin{theorem}
	For every integer $n \geq 2$, there exists a finite spectral triple
	$ (\mathcal{A}, \mathcal{H}, \mathcal{D})$ with  $\mathcal{A}=\mathbb{M}_n(\mathbb{C})$, and a faithful unital $\ast$-representation $\pi : \mathcal{A} \to
	\mathcal{B} (\mathcal{H})$, such that, for every Hermitian element $e \in
	\mathcal{A}$, 
	\[ \|[\mathcal{D}, \pi (e_{\ast})]\|_{op} = \|e_{\ast}
	\|_{op}. \]
\end{theorem}

\begin{proof}
	We construct the spectral triple explicitly.
	
	Define the spectral triple as follows:
	\begin{equation}
		\mathcal{A}=\mathbb{M}_n (\mathbb{C}), \qquad \mathcal{H}=\mathbb{C}^n
		\otimes \mathbb{C}^n, \qquad \pi (a) = \mathbb{I}_n \otimes a,
		\label{eq:representation} \quad a \in \mathcal{A}.
	\end{equation}
	Let $|i \rangle = (0, \ldots, 1, \ldots, 0)^T$ $(1 \leqslant \mathit{i}
	\leqslant n)$, and $E_{i j} = |i \rangle \langle j|$. These matrix units satisfy $E_{i j}
	|k \rangle = \delta_{j k} |i \rangle$ and $E_{i j} E_{k l} = \delta_{j k} |i
	\rangle \langle l|$. Define the flip operator
	\[ F_n = \sum_{i, j} E_{i j} \otimes E_{ji} . \]
	For the vector $|i, j \rangle = |i \rangle \otimes |j \rangle$, we have
		\[ F_n |i, j \rangle = \sum_{i', j'} E_{i' j'} |i \rangle \otimes E_{j' i'}
	|j \rangle = \sum_{i', j'} \delta_{j' i} |i' \rangle \otimes \delta_{i'
		j} |j' \rangle = |j, i \rangle . \]
	Thus, $F_n  (\xi \otimes \eta) = \eta \otimes \xi$, for any $\xi \otimes \eta \in \mathcal{H}$. One
	can check that $F_n$ is a real, self-adjoint unitary matrix,
	\[ F_n^{\dag} = \left( \sum_{i, j} E_{i j} \otimes E_{ji} \right)^{\dag} =
	\sum_{i, j} E_{ji} \otimes E_{i j} = F_n, \]
	and $F_n F_n^{\dag} = F_n^{\dag} F_n = F_n F_n = F_n^2
	=\mathbb{I}_{\mathcal{H}}$. Define the Dirac operator as
	\begin{equation}
		\mathcal{D}= \frac{1}{2} F_n, \label{eq:dirac-definition}
	\end{equation}
	and there is $\mathcal{D}^{\dag} =\mathcal{D}$.
	
	For any Hermitian element $e \in \mathcal{A}$, there is $F_n  (\mathbb{I}_n \otimes
	e) = (e \otimes \mathbb{I}_n) F_n$, and then
	\[ [\mathcal{D}, \pi (e)] = \frac{1}{2}  [F_n (\mathbb{I}_n \otimes e) - (\mathbb{I}_n \otimes
	e) F_n] = \frac{1}{2}  (e \otimes \mathbb{I}_n - \mathbb{I}_n \otimes e) F_n .
	\label{eq:commutator} \]
	Let $\lambda_1, \ldots, \lambda_n$ be the eigenvalues of $e$. Then
	\[ \| e \otimes \mathbb{I}_n - \mathbb{I}_n \otimes e \|_{op} = \max_{i, j} |
	\lambda_i - \lambda_j | = \lambda_{\max} (e) - \lambda_{\min} (e) . \]
	Since the operator $F_n$ is unitary, we have
	\[ \| [\mathcal{D}, \pi (e)] \|_{op} =  \frac{1}{2} \| e \otimes \mathbb{I}_n
	- \mathbb{I}_n \otimes e \|_{op} = \frac{\lambda_{\max} (e) -
		\lambda_{\min} (e)}{2} . \]
Since $\|[\mathcal{D}, \pi (e)]\|_{op} =\|[\mathcal{D}, \pi (e_{\ast})]\|_{op} $, from the result (\ref{eq:b-contraction}), there is
	\begin{equation}
		\| [\mathcal{D}, \pi (e_{\ast})] \|_{op} =\|e_{\ast} \|_{op}.
			\end{equation}
	
\end{proof}

For any unital representation $\pi$, any Hermitian element $e$, and any traceless operator $\Delta\rho$, we have
$\|[\mathcal{D}, \pi (e)]\|_{op} =\|[\mathcal{D}, \pi (e_{\ast})]\|_{op} $ and $\mathrm{tr} (e\Delta \rho)=\mathrm{tr} (e_{\ast} \Delta \rho)$. So one can only consider the new operator $e_*$ instead of the original Hermitian element $e$ in the calculations of Connes spectral distances.

One can also define the set
\begin{equation}
	B_* = \{e_{\ast} : e \in \mathcal{A}, e=e^{\dagger}, \|[\mathcal{D}, \pi (e)]\|_{op} \leqslant 1\},
\end{equation}
then there is $B_*\subseteq B$, and
\begin{equation}
	d (\rho_1, \rho_2) = \sup_{e \in B_*} | \mathrm{tr} (\Delta \rho
	\, e) |.
\end{equation}

Note that, in general, an element of $B_*$ may not be traceless.

\begin{theorem}
  \label{t6}
  Let $(\mathcal{A}, \mathcal{H}, \mathcal{D})$ be a spectral triple with a unital representation $\pi$. If $\mathcal{D}$ and $\pi$ satisfy the relation $\|[\mathcal{D},
  \pi (e_*)]\|_{op} = \|e_*\|_{op}$ for any Hermitian element $e \in \mathcal{A}$, then the Connes spectral distances are equal to the quantum trace distances,
  \begin{equation}
    d (\rho_1, \rho_2) = d_{\mathrm{tr}} (\rho_1, \rho_2) = \| \rho_1 - \rho_2
    \|_1 .
  \end{equation}
\end{theorem}

\begin{proof}
If $\rho_1 = \rho_2$, the result is obvious. So we only consider the case $\rho_1 \neq \rho_2$.

By Theorem \ref{qtd}, in this case, the Connes spectral distance is
\begin{eqnarray}
  d (\rho_1, \rho_2) & = & \!\!\!\!\sup_{\substack{e \in B_*\\ \|[\mathcal{D}, \pi (e)]\|_{op} \leqslant 1}}\!\!\!\!\!\! | \mathrm{tr} (\Delta \rho
  \, e) |= \sup_{\substack{e \in B_*\\ \|e\|_{op} \leqslant 1}} | \mathrm{tr} (\Delta
  \rho \, e) | \nonumber\\
  & \leqslant & d_{\mathrm{tr}} (\rho_1, \rho_2) = \| \rho_1 - \rho_2 \|_1 .
\end{eqnarray}

Next, let us explicitly construct an optimal element $e' \in B$ that achieves the maximum $\| \rho_1 - \rho_2 \|_1$.

The Hermitian operator $\Delta \rho = \rho_1 - \rho_2$ has the spectral decomposition
\begin{equation}
  \Delta \rho = \sum_i a_i |a_i \rangle \langle a_i|, \quad \sum_i a_i = 0, \quad
  a_i \neq 0.
\end{equation}
Here we only consider the nonzero eigenvalues $a_i$ of $\Delta \rho$, because the zero eigenvalues will not affect our calculations.
We have
\begin{equation}
  d_{\mathrm{tr}} (\rho_1, \rho_2) =\| \rho_1 - \rho_2 \|_1 = \mathrm{tr}
  \sqrt{\Delta \rho^{\dag} \Delta \rho} = \sum_i |a_i |.
\end{equation}

The optimal element $e'$ can be chosen as
\begin{equation}
  e' = \sum_i \frac{a_i}{|a_i |} |a_i \rangle \langle a_i|. \label{oe}
\end{equation}
Since there is $\lambda_{\max} (e')=1$, $\lambda_{\min} (e')=-1$, we have
\begin{equation}\label{ep}
	e'_{\ast} = e' - \frac{\lambda_{\max} (e') + \lambda_{\min} (e')}{2}
	\mathbb{I}_n = e' .
\end{equation}
So there are
\begin{equation}
  \|[\mathcal{D}, \pi (e')]\|_{op}=\|[\mathcal{D}, \pi (e'_{\ast})]\|_{op}=\|e'_{\ast} \|_{op} = \frac{\lambda_{\max} (e') -
  	\lambda_{\min} (e')}{2}= 1.
  \end{equation}
Hence $e'\in B$, and
\begin{equation}
	\mathrm{tr} (\Delta \rho \, e') =
  \sum_i |a_i |=\| \rho_1 - \rho_2 \|_1.
\end{equation}
Thus, the Hermitian element $e'$
is an optimal element, and we have
the Connes spectral distance
\begin{equation}
	d (\rho_1, \rho_2) = d_{\mathrm{tr}} (\rho_1, \rho_2) = \| \rho_1 - \rho_2
	\|_1 .
\end{equation}
\end{proof}

Note that, in general, the optimal element $e'$ (\ref{oe}) may not be traceless.
But if the representation $\pi$ is unital, by Lemma (\ref{tloe}),
one can choose some $c\in\mathbb{R}$ such that $e'+c\mathbb{I}$ is a traceless optimal element.

Comparing the definitions of the Connes spectral distance (\ref{csd}) and the quantum trace distance (\ref{dtr}), we observe that, in order to construct some Connes spectral distance that is equal to the quantum trace distance, the simplest case is $\|[\mathcal{D}, \pi (e)]\|_{op} = \|e \|_{op}$ for any Hermitian
elements $e$. However, this relation cannot hold in general, because $\|[\mathcal{D}, \pi (e)]\|_{op} $ remains unchanged upon adding scalar multiples of the identity matrix to $e$, but $\|e \|_{op}$ cannot remain unchanged.
Actually, one can only consider traceless Hermitian matrices in the definitions (\ref{csd}) and (\ref{dtr}), and we have the following result.

\begin{theorem}\label{thm8}
  Let $(\mathcal{A}, \mathcal{H}, \mathcal{D})$ be a spectral triple with $\mathcal{A}=\mathbb{M}_n (\mathbb{C})$, $n\geqslant 2$, and let $\pi$ be a unital $*$-representation.
  If the Dirac operator $\mathcal{D}$ satisfies the
  relation $\|[\mathcal{D}, \pi (e)]\|_{op} = \|e \|_{op}$ for any traceless Hermitian
  elements $e\in\mathcal{A}$, then $\mathcal{A}$ must be $\mathbb{M}_2 (\mathbb{C})$, and the Connes spectral distance is
\begin{equation}\label{ddr}
	d_{\mathcal{D}} (\rho_1, \rho_2) = d_{\mathrm{tr}} (\rho_1, \rho_2) = \|
  \rho_1 - \rho_2 \|_1 .
\end{equation}
\end{theorem}

\begin{proof}
	See Appendix \ref{thm8a} for a detailed proof. 
\end{proof}

Actually, using Theorem \ref{thm2}, we also have the following result.
\begin{corollary}
	Let $(\mathcal{A}, \mathcal{H}, \mathcal{D})$ be a spectral triple with $\mathcal{A}=\mathbb{M}_n (\mathbb{C})$, $n\geqslant 2$, and let $\pi$ be a unital $*$-representation.
	If the Dirac operator $\mathcal{D}$ satisfies the
	relation $\|[\mathcal{D}, \pi (e)]\|_{op} = \|e \|_{op}$ for any traceless optimal
	elements $e$ for the Connes
	spectral distances, then $\mathcal{A}$ must be $\mathbb{M}_2 (\mathbb{C})$, and the Connes spectral distance is
	\begin{equation}
		d_{\mathcal{D}} (\rho_1, \rho_2) = d_{\mathrm{tr}} (\rho_1, \rho_2) = \|
		\rho_1 - \rho_2 \|_1 .
	\end{equation}
\end{corollary}

\begin{proof}
	For any nonzero traceless Hermitian element $a \in \mathcal{A}$, as shown in the proof of Lemma \ref{d0}, we have $\|[\mathcal{D}, \pi (a)]\|_{op}\neq 0$. Define
	\[ e = \frac{a}{\|[\mathcal{D}, \pi (a)]\|_{op}}= \frac{a}{L (a)} . \]
	Then $e$ is traceless and Hermitian, and
	\[ \|[\mathcal{D}, \pi (e)]\|_{op} = 1. \]
	By Theorem~\ref{thm2}, $e$ is an optimal element for the Connes
	spectral distance between some two distinct states.
So by assumption, there is
	\[ \|[\mathcal{D}, \pi (e)]\|_{op} = \|e\|_{op} = \frac{\|a\|_{op}}{L (a)} = 1,
	\]
	which implies that, for any nonzero traceless Hermitian element $a \in
	\mathcal{A}$, there is
	\[ L (a)=\|[\mathcal{D}, \pi (a)]\|_{op} = \|a\|_{op} . \]
	This is just the condition of Theorem \ref{thm8}. So we obtain the same result.
\end{proof}

These results concerning $\mathcal{A}=\mathbb{M}_2 (\mathbb{C})$ just correspond to the case of one-qubit states.

There are also many examples
in which $\mathcal{A}$ is not a full matrix algebra $\mathbb{M}_n (\mathbb{C})$.
For example, one of the simplest cases is the two-point space.
Define the spectral triple $(\mathcal{A}, \mathcal{H}, \mathcal{D})$ with $\mathcal{A}\cong\mathbb{C}^2$ (diagonal $2 \times 2$ matrices),
$\mathcal{H}=\mathbb{C}^2$ and $D = \frac{1}{2} \left(\begin{array}{cc}
	0 & 1\\
	1 & 0
\end{array}\right)$, and $\pi(e)=e$ for $e\in\mathcal{A}$. The representation $\pi$ is also linear and unital. For any traceless Hermitian
element $e = \left(\begin{array}{cc}
	a & 0\\
	0 & - a
\end{array}\right) \in \mathcal{A}$, the Lipschitz seminorm is
\begin{equation}
\| [D, e] \|_{o p} = \left| \text{$a$} \right| = \|e\|_{op} .
\end{equation}
For any two states
\begin{equation}
\rho_1 = \left(\begin{array}{cc}
	p & 0\\
	0 & 1 - p
\end{array}\right),\qquad
\rho_2 = \left(\begin{array}{cc}
	q & 0\\
	0 & 1 - q
\end{array}\right),
\end{equation}
where $0\leqslant p,q\leqslant1$, and
\begin{equation}
\Delta\rho =\rho_1-\rho_2= \left(\begin{array}{cc}
p-q & 0\\
0 & q - p
\end{array}\right).
\end{equation}
The Connes distance is:
\begin{equation}
d (\rho_1,\rho_2) = \sup_{e \in B}[ \mathrm{tr} (\Delta \rho
\, e)] = \sup_{|a| \le 1} 2 (p - q) a = 2 |p - q| =
d_{\mathrm{tr}} (\rho_1,\rho_2) .
\end{equation}
Therefore, the spectral distance exactly matches the trace distance. The algebra $\mathcal{A}$ is commutative, so the spectral distances are unitarily invariant.

In Ref.~\cite{Kumar1}, the authors also constructed a spectral triple with a different structure, which is also related to the trace distance, and they mainly considered the cases of one-qubit states.

\section{Connes spectral distances between one-qubit states}\label{sec6}
In this section, we will study some concrete examples for one-qubit states. These spectral triples can help us study the geometric structure of quantum states.

Following a construction similar to that in Ref.~\cite{LX}, we define a spectral triple as follows.
\begin{equation}
	\mathcal{A}= \{ |i \rangle
	\langle j| \}\cong\mathbb{M}_2 (\mathbb{C}) ,\qquad \mathcal{H}_4  =\mathbb{C}^2 \otimes F\cong\mathbb{C}^4,\qquad F = \{
|0 \rangle, |1 \rangle \},
\end{equation}
and $\pi (a) =\mathbb{I}_2 \otimes a = \left(
\begin{array}{cc}
  a & 0\\
  0 & a
\end{array} \right)$, $\forall a \in \mathcal{A}$. The Dirac operator is
\begin{equation}\label{dd4}
  \mathcal{D}_4  =
  \frac{1}{4}  (s_1\sigma_1 \otimes \sigma_2 + s_2\sigma_2 \otimes \sigma_3 +
  s_3\sigma_3 \otimes \sigma_1),
\end{equation}
where $\sigma_i$ are the Pauli matrices, and $s_i \in \{- 1, 1\}$. In this case, the representation
$\pi$ is linear and unital.

In general, the density matrix $\rho$ for a qubit can be expressed as
\cite{Nielsen}
\begin{equation}
  \label{rr} \rho = \frac{I + \vec{r} \cdot \vec{\sigma}}{2},
\end{equation}
where the real vector $\vec{r}$ is the so-called Bloch vector, $|
\vec{r} | \leqslant 1$, and $\vec{\sigma} = (\sigma_1, \sigma_2, \sigma_3)$.

\begin{lemma}\label{l9}
  In the spectral triple $(\mathbb{M}_2 (\mathbb{C}), \mathbb{C}^4,
  \mathcal{D}_4)$, for any traceless Hermitian matrix $e \in \mathbb{M}_2
  (\mathbb{C})$, there is $\|[\mathcal{D}_4, \pi (e)]\|_{op} = \|e\|_{op}$.
\end{lemma}

\begin{proof}
	Any $2 \times 2$ traceless Hermitian matrix $e$ can be
expressed by the Pauli matrices as
\begin{equation}
	e = \vec{e} \cdot \vec{\sigma} = e_1 \sigma_1 + e_2 \sigma_2 + e_3 \sigma_3,
\end{equation}
where $\vec{e} = (e_1, e_2, e_3)$ is a real vector.
Using the following identities for the Pauli matrices
\begin{equation}
	\sigma_i \sigma_j = \delta_{i j} \mathbb{I}+ \text{i} \varepsilon_{i j k}
	\sigma_k,
\end{equation}
\begin{equation}
	(\vec{a} \cdot \vec{\sigma})  (\vec{b} \cdot \vec{\sigma}) = (\vec{a} \cdot
	\vec{b}) \mathbb{I}_2 + \text{i}  (\vec{a} \times \vec{b}) \cdot
	\vec{\sigma},
\end{equation}
\begin{equation}
	[\vec{a} \cdot \vec{\sigma}, \vec{b} \cdot \vec{\sigma}] = 2 \text{i}
	(\vec{a} \times \vec{b}) \cdot \vec{\sigma},
\end{equation}
after some straightforward calculations, one can obtain
\begin{equation}
	e^{\dag} e = | \vec{e}  |^2 \cdot \mathbb{I}_2,
\end{equation}
\begin{equation}
	[\mathcal{D}_4, \mathbb{I}_2 \otimes e]^{\dag} [\mathcal{D}_4, \mathbb{I}_2
	\otimes e] = \frac{1}{2} (| \vec{e}  |^2 \cdot \mathbb{I}_4 - \tilde{e} \otimes e).
\end{equation}
where $\tilde{e} = s_2 s_3 e_2 \sigma_1 + s_3 s_1 e_3 \sigma_2 + s_1 s_2 e_1 \sigma_3$.
The eigenvalues of $e = \vec{e} \cdot \vec{\sigma}$ and $\tilde{e}$ are $\{ | \vec{e} |,- | \vec{e} | \}$.
Since the eigenvalues of a tensor product of two matrices $A \otimes B$ are the
products of the eigenvalues of $A$ and the eigenvalues of $B$, the eigenvalues of $\tilde{e} \otimes e$ are $\{| \vec{e} |^2, | \vec{e} |^2, - | \vec{e} |^2,
- | \vec{e} |^2 \}$.
So the eigenvalues of $\frac{1}{2}  (| \vec{e} |^2 \mathbb{I}_4 - \tilde{e} \otimes e)$
are $\{0, 0, | \vec{e} |^2, | \vec{e} |^2 \}$.

Therefore, we have
\begin{equation}
  \|[\mathcal{D}_4, \pi (e)]\|_{op} = | \vec{e} | = \|e\|_{op} .
\end{equation}
\end{proof}

Consider any two one-qubit states $\rho_1, \rho_2$, with corresponding Bloch
vectors $\vec{r}_1$, $\vec{r}_2$, respectively.
By Lemma \ref{l9} and Theorem \ref{thm8}, one can obtain the Connes spectral distance and quantum trace distance \cite{Nielsen},
\begin{equation}
	d (\rho_1, \rho_2) = d_{\mathrm{tr}} (\rho_1, \rho_2) = |\vec{r}_1-\vec{r}_2| =| \Delta \vec{r} | = \mathrm{tr} (\Delta
	\rho e_o) .
\end{equation}
In this case, there is
\begin{equation}
	\Delta \rho =\rho_1-\rho_2= \frac{1}{2} (\vec{r}_1-\vec{r}_2) \cdot \vec{\sigma}= \frac{1}{2} \Delta \vec{r} \cdot \vec{\sigma},
\end{equation}
and the traceless optimal element can be chosen as
\begin{equation}
	e_o = \frac{\Delta \vec{r}}{| \Delta \vec{r} |} \cdot \vec{\sigma} =
	\frac{\Delta \rho}{\| \Delta \rho \|_{op}}.
\end{equation}
The optimal element set $B_o$ is
\begin{equation}
	B_o = \{ \vec{n} \cdot \vec{\sigma} : | \vec{n} | = 1\} .
\end{equation}
Since a unitary transformation of $\vec{n} \cdot \vec{\sigma}$ is just some
rotation of the unit vector $\vec{n}$, the set $B_o$ is unitarily invariant.

Therefore, we have the following result.
\begin{corollary}
  In the spectral triple $(\mathbb{M}_2 (\mathbb{C}), \mathbb{C}^4,
  \mathcal{D}_4)$, the Connes spectral distance between any two one-qubit states
  $\rho_1, \rho_2$ is equal to both their quantum trace distance and the Euclidean
  distance between the Bloch vectors,
  \begin{equation}
    d (\rho_1, \rho_2) = d_{\mathrm{tr}} (\rho_1, \rho_2) = | \vec{r}_1 - \vec{r}_2 |,
  \end{equation}
  where $\vec{r}_1$, $\vec{r}_2$ are the corresponding Bloch vectors of
  $\rho_1, \rho_2$, respectively. The corresponding optimal element $e_o$ is
  \begin{equation}
    e_o = \frac{\Delta \vec{r}}{| \Delta \vec{r} |} \cdot \vec{\sigma}= \frac{\Delta \rho}{\| \Delta \rho \|_{op}} .
  \end{equation}
\end{corollary}

It is easy to see that, for any unitary matrix $U \in \mathbb{M}_2
(\mathbb{C})$, the Dirac operator $\mathcal{D}_4' = \pi (U) \mathcal{D}_4 \pi
(U^{\dag}) = (\mathbb{I}_2 \otimes U) \mathcal{D}_4  (\mathbb{I}_2 \otimes
U^{\dag})$ also satisfies the property $\|[\mathcal{D}'_4, \pi (e)]\|_{op} =
\|e\|_{op}$ for any traceless Hermitian matrix $e$,
\begin{eqnarray}
  \|[\mathcal{D}'_4, \pi (e)]\|_{op} & = & \|[\pi (U)\mathcal{D}_4 \pi
  (U^{\dag}), \pi (e)]\|_{op}= \|[\mathcal{D}_4, \pi (U^{\dag}) \pi (e) \pi (U)]\|_{op}\nonumber\\
  & = & \|[\mathcal{D}_4, \pi (U^{\dag} e U)]\|_{op}= \|U^{\dag} e U\|_{op}= \|e\|_{op} .
\end{eqnarray}
For example, there are $2 \times 2$ unitary matrices $U_+, U_-$ which can
permute the Pauli matrices,
\begin{equation}
  U_+ = \frac{1}{2} \left( \begin{array}{cc}
    1 - i & - 1 - i\\
    1 - i & 1 + i
  \end{array} \right), \quad U_- = \frac{1}{2} \left( \begin{array}{cc}
    1 + i & 1 + i\\
    - 1 + i & 1 - i
  \end{array} \right),
\end{equation}
and
\begin{equation}
  U_+ \sigma_1 U_+^{\dag} = \sigma_2, \quad U_+ \sigma_2 U_+^{\dag} =
  \sigma_3, \quad U_+ \sigma_3 U_+^{\dag} = \sigma_1 ;
\end{equation}
\begin{equation}
  U_- \sigma_3 U_-^{\dag} = \sigma_2, \quad U_- \sigma_2 U_-^{\dag} =
  \sigma_1, \quad U_- \sigma_1 U_-^{\dag} = \sigma_3 .
\end{equation}
Using the unitary Hadamard matrix
\[ H = \frac{1}{\sqrt{2}} \left( \begin{array}{cc}
	1 & 1\\
	1 & - 1
\end{array} \right), \]
we also have
\begin{equation}
	H \sigma_1 H^{\dagger} = \sigma_3,\quad
	H \sigma_2 H^{\dagger} = - \sigma_2,\quad
	H \sigma_3 H^{\dagger} = \sigma_1 .
\end{equation}
Therefore, using the Dirac operator $\mathcal{D}_4$ (\ref{dd4}) and the above unitary matrices $U_{\pm}$ and $H$, one can obtain the following result.
\begin{lemma}\label{l10}
  For any traceless Hermitian matrix $e \in \mathbb{M}_2 (\mathbb{C})$, the
  Dirac operator
  \begin{equation}\label{d4p}
    \mathcal{D}_4' = \frac{1}{4}  \sum_{i = 1}^3 s_i \sigma_i \otimes
    \sigma_{i'}
  \end{equation}
  satisfies the relation $\|[\mathcal{D}_4', \pi (e)]\|_{op} = \|e\|_{op}$,
  where $\sigma_i$ are the Pauli matrices, $s_i \in \{- 1, 1\}$, and $(\sigma_{1'},\sigma_{2'},\sigma_{3'})  $ is an arbitrary permutation of $\{\sigma_1, \sigma_2, \sigma_3\}$.
\end{lemma}

\begin{proof}
It is easy to see that all cases of $\mathcal{D}_4'$ can be obtained from the Dirac operator $\mathcal{D}_4$ (\ref{dd4}) by applying suitable unitary transformations constructed from the unitary matrices $U_{\pm}$ and $H$.

	For example, from the definition (\ref{dd4}) and Lemma \ref{l9}, the
	following Dirac operator
	\begin{equation}
		D_1 = \frac{1}{4}  (\sigma_1 \otimes \sigma_2 + \sigma_2 \otimes \sigma_3
		+ \sigma_3 \otimes \sigma_1)
	\end{equation}
	satisfies the relation $\|[D_1, \pi (e)]\|_{op} = \|e\|_{op}$ for any
	traceless Hermitian matrix $e \in \mathbb{M}_2 (\mathbb{C})$. Consider the
	unitary matrix
	\[ U = U_- H, \]
	and we have
	\begin{eqnarray*}
		D_2 & = & \pi (U) D_1 \pi (U^{\dag})= (\mathbb{I}_2 \otimes U_- H) D_1  (\mathbb{I}_2 \otimes H^{\dag}
		U_-^{\dag})\\
		& = & \frac{1}{4}  (- \sigma_1 \otimes \sigma_1 + \sigma_2 \otimes
		\sigma_3 + \sigma_3 \otimes \sigma_2),
	\end{eqnarray*}
	which also satisfies the relation $\|[D_2, \pi (e)]\|_{op} = \|e\|_{op}$ for any
	traceless Hermitian matrix $e \in \mathbb{M}_2 (\mathbb{C})$.
Other cases are similar.
\end{proof}

\begin{theorem}
  Consider the Dirac operator,
  \begin{equation}
    \mathcal{D}_{4^n} = \frac{1}{4}  \sum_{i = 1}^3 s_i \sigma_{i_n, i_n'}
    \otimes \cdots \otimes \sigma_{i_1, i_1'},
  \end{equation}
  where $\sigma_{i_k,
  i_k'} = \sigma_{i_k} \otimes \sigma_{i_k'}$, $s_i \in \{- 1, 1\}$, and $(\sigma_{1_k},\sigma_{2_k},\sigma_{3_k})$ and $(\sigma_{1_k'},\sigma_{2_k'},\sigma_{3_k'})$ are arbitrary permutations of $\{\sigma_1, \sigma_2, \sigma_3\}$. Then for any $2 \times 2$
  traceless Hermitian matrix $e \in \mathbb{M}_2 (\mathbb{C})$, we have
  \begin{equation}
    \|[\mathcal{D}_{4^n}, \pi_n (e)]\|_{op} = \|e\|_{op},
  \end{equation}
  where $\pi_n (e) =\mathbb{I}_{4^{n - 1}} \otimes \mathbb{I}_2 \otimes e$.
\end{theorem}

\begin{proof}
	Using the unitary matrices $U_{\pm}$ and $H$, one can permute the Pauli
	matrices. So we only need to consider the simplest case
	\[ \mathcal{D}_{4^n} = \frac{1}{4}  \sum_{i = 1}^3 s_i  \underbrace{\sigma_{i, i}
		\otimes \cdots \otimes \sigma_{i, i}}_n := \frac{1}{4}  \sum_{i = 1}^3 s_i 
	\sigma_{i, i}^{\otimes n}. \]
	
	When $n = 1$, this is just the case in Lemma \ref{l10}.
	
	When $n = 2$, we have
	\begin{equation}
		\mathcal{D}_{16} = \frac{1}{4}  \sum_{i = 1}^3 s_i  \sigma_{i, i} \otimes
		\sigma_{i, i},
	\end{equation}
	and
	\begin{eqnarray}
		\|[\mathcal{D}_{16}, \pi_2 (e)]\|_{op} & = & \left\| \left[ \frac{1}{4} 
		\sum_{i = 1}^3 s_i \sigma_{i, i} \otimes \sigma_{i, i}, \mathbb{I}_4 \otimes
		\pi_1 (e) \right] \right\|_{op} \nonumber\\
		& = & \left\| \frac{1}{4} \sum_{i = 1}^3 s_i  \sigma_{i, i} \otimes [\sigma_{i,
			i}, \pi_1 (e)] \right\|_{op} . 
	\end{eqnarray}
Using the following $4 \times 4$ unitary matrix
	\begin{equation}\label{uu}
		U = \frac{1}{\sqrt{2}} \left( \begin{array}{cccc}
			1 & 0 & 0 & - 1\\
			0 & 1 & - 1 & 0\\
			0 & 1 & 1 & 0\\
			1 & 0 & 0 & 1
		\end{array} \right),
	\end{equation}
we have
	\begin{eqnarray*}
		\|[\mathcal{D}_{16}, \pi_2 (e)]\|_{op} & = & \left\| (U \otimes
		\mathbb{I}_4) \left( \frac{1}{4}  \sum_{i = 1}^3 s_i \sigma_{i, i} \otimes
		[\sigma_{i, i}, \pi_1 (e)] \right) (U^{\dag} \otimes \mathbb{I}_4)
		\right\|_{op} \\
		& = & \left\| \left( \begin{array}{cccc}
			{}[D_1, \pi_1 (e)] & 0 & 0 & 0\\
			0 & [D_2, \pi_1 (e)] & 0 & 0\\
			0 & 0 & [D_3, \pi_1 (e)] & 0\\
			0 & 0 & 0 & [D_4, \pi_1 (e)]
		\end{array} \right) \right\|_{op}, \nonumber
	\end{eqnarray*}
	where
	\begin{equation}
		D_k = \frac{1}{4}  \sum_{i = 1}^3 s_i \sigma_{i, i} = \frac{1}{4}  \sum_{i =
			1}^3 s_i \sigma_i \otimes \sigma_i, \quad k = 1, 2, 3, 4.
	\end{equation}
Note that the signs of the coefficients may differ among the operators $D_k$,
but by Lemma \ref{l10}, they all
	satisfy the same relation as the Dirac operator $\mathcal{D}_4'$ (\ref{d4p}),
	\begin{equation}
		\|[D_k, \pi_1 (e)]\|_{op} = \|[\mathcal{D}_4', \pi_1 (e)]\|_{op} =
		\|e\|_{op}, \quad k = 1, 2, 3, 4.
	\end{equation}
	So we have
	\begin{equation}
		\|[\mathcal{D}_{16}, \pi_2 (e)]\|_{op}  = \|e\|_{op} .
	\end{equation}
	
	Now let us proceed by induction. Assume that the statement holds for some $n \in
	\mathbb{N}^+$, and
	\begin{equation}\label{d4n}
		\mathcal{D}_{4^n} = \frac{1}{4}  \sum_{i = 1}^3 s_i \sigma_{i, i}^{\otimes n},
		\qquad \|[\mathcal{D}_{4^n}, \pi_n (e)]\|_{op} = \|e\|_{op} .
	\end{equation}
	Consider the case of $n + 1$, and we have
	\begin{equation}
		\mathcal{D}_{4^{n + 1}} = \frac{1}{4}  \sum_{i = 1}^3 s_i \sigma_{i, i}
		\otimes \underbrace{\sigma_{i, i} \otimes \cdots \otimes \sigma_{i, i}}_n
		:= \frac{1}{4}  \sum_{i = 1}^3 s_i \sigma_{i, i} \otimes \sigma_{i,
			i}^{\otimes n},
	\end{equation}
	and
	\begin{eqnarray}
		\|[\mathcal{D}_{4^{n + 1}}, \pi_{n + 1} (e)]\|_{op} & = & \left\| \left[
		\frac{1}{4}  \sum_{i = 1}^3 s_i \sigma_{i, i} \otimes \sigma_{i, i}^{\otimes
			n}, \mathbb{I}_4 \otimes \pi_n (e) \right] \right\|_{op} \nonumber\\
		& = & \left\| \frac{1}{4} \sum_{i = 1}^3 s_i \sigma_{i, i} \otimes
		[\sigma_{i, i}^{\otimes n}, \pi_n (e)] \right\|_{op} . 
	\end{eqnarray}
Similarly, using the unitary matrix $U$ (\ref{uu}), we have
	\begin{eqnarray*}
		\|[\mathcal{D}_{4^{n + 1}}, \pi_{n + 1} (e)]\|_{op}\!\!\!\! & = & \!\!\!\!\left\| (U
		\otimes \mathbb{I}_{4^n}) \left( \frac{1}{4}  \sum_{i = 1}^3 s_i \sigma_{i, i}
		\otimes [\sigma_{i, i}^{\otimes n}, \pi_n (e)] \right) (U^{\dag} \otimes
		\mathbb{I}_{4^n}) \right\|_{op} \\
		& = &\!\!\!\! \left\|\! \left(\! \begin{array}{cccc}
			{}[D_1', \pi_n (e)] & 0 & 0 & 0\\
			0 & [D_2', \pi_n (e)] & 0 & 0\\
			0 & 0 & [D_3', \pi_n (e)] & 0\\
			0 & 0 & 0 & [D_4', \pi_n (e)]
		\end{array} \!\right)\! \right\|_{op}, \nonumber
	\end{eqnarray*}
	where
	\begin{equation}
		D_k' = \frac{1}{4}  \sum_{i = 1}^3 s_i \sigma_{i, i}^{\otimes n} = \frac{1}{4}
		\sum_{i = 1}^3 s_i \underbrace{\sigma_{i, i} \otimes \cdots \otimes
			\sigma_{i, i}}_n, \quad k = 1, 2, 3, 4.
	\end{equation}
	Note that in the above expression, $D_k'$ and $\mathcal{D}_{4^n}$ may have different signs $s_i$ of
	some terms. However, these sign changes do not affect the result, provided that the induction hypothesis applies to all sign combinations.
	They all
	satisfy the same relation,
	\begin{equation}
		\|[D_k', \pi_n (e)]\|_{op} = \|[\mathcal{D}_{4^n}, \pi_n (e)]\|_{op} =
		\|e\|_{op}, \quad k = 1, 2, 3, 4.
	\end{equation}
	So we have
	\begin{equation}
		\|[\mathcal{D}_{4^{n + 1}}, \pi_{n + 1} (e)]\|_{op} = \|[D_k', \pi_n
		(e)]\|_{op} = \|e\|_{op} 
	\end{equation}
for the case $n+1$. This proves the result.
\end{proof}

For example, consider the following Dirac operator
\begin{equation}
	\mathcal{D}_{16} = \frac{1}{4}  \sum_{i = 1}^3 \sigma_i \otimes \sigma_i
	\otimes \sigma_i \otimes \sigma_i, \quad {\pi_2}  (e) =\mathbb{I}_8 \otimes
	e.
\end{equation}
Then, for every $2 \times 2$
traceless Hermitian matrix $e \in \mathbb{M}_2 (\mathbb{C})$, we have
\begin{equation}
	\|[\mathcal{D}_{16}, \pi_2 (e)]\|_{op} = \|e\|_{op} .
\end{equation}

\begin{lemma}\label{lem11}
Let $A = \sum_{i, j} S_i \otimes T_j$, where $S_i$ are $k
  \times k$ matrices and $T_j$ are $l \times l$ matrices. Define a new matrix $A'$ by 
  \begin{equation}
    A' = \sum_{i, j} S_i \otimes \tilde{M} \otimes T_j,
  \end{equation}
  where $M$ is any non-zero $n \times n$ matrix, and $\tilde{M} =
  \frac{M}{\|M\|_{op}}$, then there is $\|A\|_{op} = \|A' \|_{op}$.
\end{lemma}

\begin{proof}
	There exists some $knl \times knl$ unitary permutation
matrix $P$,
\begin{equation}
  P (S_i \otimes \tilde{M} \otimes T_j) P^{\dag} = \tilde{M} \otimes S_i
  \otimes T_j .
\end{equation}
Obviously, there is $\| \tilde{M} \|_{op} = 1$, and
\begin{eqnarray*}
  \|A' \|_{op} & = & \left\| \sum_{i, j} S_i \otimes \tilde{M} \otimes T_j
  \right\|_{op} = \left\| P \left( \sum_{i, j} S_i \otimes \tilde{M} \otimes T_j
  \right) P^{\dag} \right\|_{op} \nonumber\\
  & = & \left\| \tilde{M} \otimes \sum_{i, j} S_i \otimes T_j \right\|_{op}
  = \| \tilde{M} \otimes A\|_{op} = \| \tilde{M} \|_{op} \cdot \|A\|_{op}
  = \|A\|_{op} . 
\end{eqnarray*}
\end{proof}

\begin{theorem}
	Consider a spectral triple $(\mathbb{M}_k (\mathbb{C}), \mathcal{H},
	\mathcal{D})$ with $\mathcal{H}=\mathbb{C}^{mk}$ and $\mathcal{D}= \sum_{i,
		j} S_i \otimes T_j$, where $S_i$ are $m' \times m'$ matrices, and $T_j$ are $k' \times k'$ matrices with $k' = k l$ for some integer $l \geqslant
	1$, and $mk = m' k'$. An element $a \in \mathbb{M}_k (\mathbb{C})$ acts on
	$\mathcal{H}$ through the diagonal representation $\pi (a) =\mathbb{I}_m
	\otimes a$. Then the corresponding spectral triple $(\mathbb{M}_k
	(\mathbb{C}), \mathcal{H}', \mathcal{D}')$ with
	\begin{equation}
		\mathcal{H}' =\mathbb{C}^n \otimes \mathcal{H}, \qquad \mathcal{D}' =
		\sum_{i, j} S_i \otimes \tilde{M} \otimes T_j
	\end{equation}
	has the same metric as $(\mathbb{M}_k (\mathbb{C}), \mathcal{H},
	\mathcal{D})$, where $M$ is any non-zero $n \times n$ Hermitian matrix,
	$\tilde{M} = \frac{M}{\|M\|_{op}}$, and an element $a \in \mathbb{M}_k
	(\mathbb{C})$ acts on $\mathcal{H}'$ through the diagonal representation
	$\pi' (a) =\mathbb{I}_{m n} \otimes a$.
\end{theorem}

\begin{proof}
	Since $k' = k l$ and $mk = m' k'$, we have $m = m' l$. Using the result in
	Lemma \ref{lem11}, we have
	\begin{eqnarray}
		\|[\mathcal{D}', \pi' (e)]\|_{op} & = & \left\| \left[ \sum_{i, j} S_i
		\otimes \tilde{M} \otimes T_j, \mathbb{I}_{m n} \otimes e \right]
		\right\|_{op} \nonumber\\
		& = & \left\| \left[ \sum_{i, j} S_i \otimes \tilde{M} \otimes T_j,
		\mathbb{I}_{m' n} \otimes \mathbb{I}_l \otimes e \right] \right\|_{op}
		\nonumber\\
		& = & \left\| \sum_{i, j} S_i \otimes \tilde{M} \otimes [T_j,
		\mathbb{I}_l \otimes e] \right\|_{op} = \left\| \sum_{i, j} S_i \otimes [T_j, \mathbb{I}_l \otimes e]
		\right\|_{op} \nonumber\\
		& = & \left\| \left[ \sum_{i, j} S_i \otimes T_j, \mathbb{I}_{m'} \otimes
		\mathbb{I}_l \otimes e \right] \right\|_{op} = \left\| \left[ \sum_{i, j}
		S_i \otimes T_j, \mathbb{I}_m \otimes e \right] \right\|_{op} \nonumber\\
		& = & \|[\mathcal{D}, \pi (e)]\|_{op} . 
	\end{eqnarray}
	So the spectral triples $(\mathbb{M}_k (\mathbb{C}), \mathcal{H},
	\mathcal{D})$ and $(\mathbb{M}_k (\mathbb{C}), \mathcal{H}', \mathcal{D}')$
	have the same ball condition,
	\begin{equation}
		B' \equiv \{e \in \mathcal{A}: \|[\mathcal{D}', \pi' (e)]\|_{op} \leqslant
		1\} = B \equiv \{e \in \mathcal{A}: \|[\mathcal{D}, \pi (e)]\|_{op}
		\leqslant 1\} .
	\end{equation}
	Therefore, the Connes spectral distance between the states $\rho_1, \rho_2
	\in \mathbb{M}_k (\mathbb{C})$ in the spectral triple $(\mathbb{M}_k
	(\mathbb{C}), \mathcal{H}', \mathcal{D}')$ is equal to that of the
	corresponding states in $(\mathbb{M}_k (\mathbb{C}), \mathcal{H},
	\mathcal{D})$,
	\begin{equation}
		d_{\mathcal{D}} (\rho_1, \rho_2) = \sup_{e \in B} | \mathrm{tr}_F (\Delta
		\rho \hspace{0.17em} e) | = \sup_{e \in B'} | \mathrm{tr}_F (\Delta \rho
		\hspace{0.17em} e) | = d_{\mathcal{D}'} (\rho_1, \rho_2) .
	\end{equation}
\end{proof}

For example, it is straightforward to verify that, for the following spectral triple, the spectral distances between one-qubit states coincide with both the quantum trace distances and the Euclidean distances between the corresponding Bloch vectors,
\begin{equation}
  \mathcal{A}=\mathbb{M}_2 (\mathbb{C}), \quad \mathcal{H}_8 =\mathbb{C}^8,
  \quad \mathcal{D}_8 = \frac{1}{4}  (\sigma_2 \otimes \tilde{M} \otimes
  \sigma_1 + \sigma_3 \otimes \tilde{M} \otimes \sigma_2 + \sigma_1 \otimes
  \tilde{M} \otimes \sigma_3),
\end{equation}
where $M$ is a non-zero $2\times 2$ Hermitian matrix, and
$\tilde{M} = \frac{M}{\|M\|_{op}}$.

The above concrete models can help us construct other types of spectral triples with unitarily invariant metrics. These results can also help us understand the relations between the structure of Dirac operators and the properties of Connes spectral distance.

Furthermore, one can also use a similar method to construct spectral triples corresponding to higher-dimensional quantum states. But usually the calculations will be more cumbersome.

\section{Discussions and conclusions}\label{sec7}

In this paper, we study several aspects of the unitary invariance of Connes spectral distances between quantum states. We mainly focus on spectral triples with matrix algebras acting on finite-dimensional Hilbert spaces via some unital representations. Furthermore, we also assume that the spectral distances are all finite.

We derive some elementary properties of the Connes spectral distances and the optimal elements. We prove that there is a direct correspondence between the Lipschitz unit sphere and the set of optimal elements of Connes spectral distances. A Hermitian element is an optimal element for the Connes spectral distance between some two distinct states if and only if the Lipschitz seminorm equals $1$.

The unitary invariance of the spectral distances is equivalent to the unitary invariance of the Lipschitz seminorm.
If we assume the representation to be unital, then the spectral distances are also unitarily invariant under some unitary transformations of the Dirac operators.

One of the most important unitarily invariant quantum distance measures is the quantum trace distance.
We prove that for every finite-dimensional matrix algebra $\mathcal{A}=\mathbb{M}_n(\mathbb{C})$, there exist finite-dimensional spectral triples in which
the Connes spectral distances coincide with the normal quantum trace distances between quantum states.

One of the simplest cases is that the Lipschitz seminorm equals the operator norm for all traceless Hermitian elements.
We prove that, in this case, the only possible choice of the matrix algebra $\mathcal{A}$ is $\mathbb{M}_2(\mathbb{C})$. This just corresponds to the case of one-qubit states.
We also explicitly construct some spectral triples for one-qubit states in which the Connes spectral distances are exactly the quantum trace distances.

These results are helpful for us to better understand the unitary invariance of Connes spectral distance and the relationships among Connes spectral distance, quantum trace distance and other quantum distance measures.
These concrete examples are significant for studies of geometric structures of finite spectral triples and mathematical relations of qubits and other quantum states in the framework of noncommutative geometry. These results may also be useful in studies of Connes spectral distances in other types of spectral triples.

\section*{Acknowledgments}
This work is partly supported by the
Guangdong Basic and Applied Basic Research Foundation (Grant No.
2024A1515010380), the 2024 Guangdong Province Education Science Planning
Project (Higher Education Special) (No. 2024GXJK455), and the 2024 Guangdong
Higher Education Teaching Reform Project.

\appendix

\section{Proof of Theorem \ref{thm2}}\label{thm2a}
\begin{proof}
	We give an explicit construction of the two density matrices $\rho_1$ and
	$\rho_2$.
	
	Define the linear map
	\begin{equation}
		T (a) = i [D, \pi (a)].
	\end{equation}
	Then we have $T (\mathbb{I}_n) = i [D, \pi (\mathbb{I}_n)] = 0$, $T
	(a^{\dag}) = T (a)^{\dag}, a \in \mathcal{A}$. In particular, $T (a)$ is
	Hermitian whenever $a$ is Hermitian. Moreover, we have
	\begin{equation*}
		\|T (a)\|_{o p} = \|[D, \pi (a)]\|_{o p} = L_D (a) .
	\end{equation*}
	Since $\mathcal{H}$ is finite-dimensional, the operator norm is attained at an eigenvector. For any Hermitian element $e$ satisfying $\|T(e)\|_{op} = L_D (e) = 1$, by the definition of operator norm,
	there exist a unit vector $|  v \rangle  \in \mathcal{H}$ and a sign
	$\varepsilon \in \{- 1, 1\}$ such that
	\begin{equation}
		T (e)  |  v \rangle  = \varepsilon |  v \rangle , \qquad \langle v |  v
		\rangle = 1. \label{eq:extremal-eigenvector}
	\end{equation}
	
	Let $E_{i j} = |i \rangle \langle j|$ denote the standard matrix units of
	$\mathbb{M}_n (\mathbb{C})$, and define a matrix $h = (h_{i j}) \in \mathbb{M}_n
	(\mathbb{C})$ by
	\begin{equation}
		h_{i j} = \varepsilon \langle v|T (E_{ji}) |v \rangle, \qquad 1 \leqslant i, j
		\leqslant n. \label{eq:h-construction}
	\end{equation}
	Since $E_{i j}^{\dag} = E_{ji}$, we have
	\begin{equation*}
		\overline{h_{i j}} = \varepsilon \overline{\langle v|T (E_{ji}) |v
			\rangle} = \varepsilon \langle v|T (E_{ji})^{\dag} |v \rangle = \varepsilon \langle v|T (E_{i j}) |v \rangle = h_{ji} . 
	\end{equation*}
	Therefore, we have $h = h^{\dag}$. Moreover, there is
	\begin{equation*}
		\mathrm{Tr} (h) = \sum_{i = 1}^n h_{ii}  = \varepsilon \sum_{i = 1}^n \langle v|T (E_{i i}) |v \rangle
		= \varepsilon \left\langle v \left| T \left( \sum_{i = 1}^n E_{ii}
		\right) \right| v \right\rangle  = \varepsilon \langle v | T (\mathbb{I}_n) | v \rangle = 0. 
	\end{equation*}
	Thus, $h$ is a traceless Hermitian matrix.
	
	Next, for any $a \in \mathbb{M}_n (\mathbb{C})$, there is
	$ a = \sum_{i, j = 1}^n a_{i j} E_{i j} $, and
	\begin{equation*}
		\mathrm{Tr} (ha) = \sum_{i, j = 1}^n h_{ji} a_{i j} = \varepsilon \sum_{i, j = 1}^n a_{i j}  \langle v | T (E_{i j}) | v
		\rangle = \varepsilon \langle v | T (a) | v \rangle . 
	\end{equation*}
	So for every Hermitian $a \in \mathbb{M}_n (\mathbb{C})$, we obtain
	\begin{equation}
		| \mathrm{Tr} (ha) | = | \langle v | T (a) | v \rangle | \leqslant \|T (a)\|_{o p} = L_D (a) . 
	\end{equation}
	On the other hand, there is
	\begin{equation}
		\mathrm{Tr} (he) = \varepsilon \langle v | T (e) | v \rangle  = \varepsilon^2  \langle v|v \rangle  = 1. \label{he} 
	\end{equation}
	Therefore, we have constructed a nonzero traceless Hermitian matrix $h$ satisfying
	\[ | \mathrm{Tr} (ha) | \leqslant L_D (a) \]
	for every Hermitian $a$, with equality at the prescribed element $e$. So there
	is
	\begin{equation}
		\sup_{\substack{a = a^{\dag}\\ L_D (a) \leqslant 1}} | \mathrm{Tr} (ha) | = 1.
		\label{eq:f-e}
	\end{equation}
	
	Since $h$ is a nonzero traceless Hermitian matrix, its spectrum contains
	both positive and negative eigenvalues. The spectral decomposition
	of $h$ can be written as
	\begin{equation}
		h = \sum_{j = 1} \lambda_j |u_j \rangle \langle u_j |,\quad \lambda_j\in \mathbb{R}.
	\end{equation}
	Define
	\begin{equation*}
		h_+ = \sum_{\lambda_j > 0} \lambda_j |u_j \rangle \langle u_j |
		\geqslant 0, \qquad
		h_-  = \sum_{\lambda_j < 0} (- \lambda_j) |u_j \rangle \langle u_j |
		\geqslant 0, 
	\end{equation*}
	and we have $ h = h_+ - h_-$ and $h_+ h_- = 0$. Since $\mathrm{Tr} (h) = 0$, there is
	\begin{equation}
		t := \mathrm{Tr} (h_+) = \mathrm{Tr} (h_-) = \frac{1}{2} \|h\|_1 .
	\end{equation}
	Since $h \neq 0$, both $h_+$ and $h_-$ are nonzero, and $t > 0$. Now define
	\begin{equation}
		\rho_1 = \frac{h_+}{t}, \qquad \rho_2 = \frac{h_-}{t} .
	\end{equation}
	It is easy to see that $\rho_1, \rho_2 \geqslant 0$ and $\mathrm{Tr} (\rho_1) =
	\mathrm{Tr} (\rho_2) = 1$. Thus $\rho_1$ and $\rho_2$ are density matrices.
	Since $h_+ h_- = 0$, we have $\rho_1 \rho_2 = 0$, and $\rho_1 \neq \rho_2$.
	So there is $h = t (\rho_1 - \rho_2)$, and
	\begin{equation}
		\Delta \rho = \rho_1 - \rho_2 = \frac{h}{t} .
		\label{eq:functional-state-difference}
	\end{equation}
	Using the result \eqref{eq:f-e}, the Connes spectral distance between the
	above states $\rho_1$ and $\rho_2$ is
	\begin{equation}
		d_D (\rho_1, \rho_2) = \sup_{\substack{a = a^{\dag}\\ L_D (a) \leqslant 1}} |
		\mathrm{Tr} ((\rho_1 - \rho_2) a) | = \frac{1}{t} \sup_{\substack{a =
				a^{\dag}\\ L_D (a) \leqslant 1}} | \mathrm{Tr} (ha) | = \frac{1}{t} . 
	\end{equation}
	On the other hand, by \eqref{he} and \eqref{eq:functional-state-difference}, we have
	\begin{equation}
		| \mathrm{Tr} ((\rho_1 - \rho_2) e) | = \frac{1}{t} | \mathrm{Tr} (he) |
		= \frac{1}{t} . 
	\end{equation}
	Since $L_D (e)=\|[D, \pi (e)]\|_{o p} = 1$, the given Hermitian element $e$ is an optimal
	element for the Connes spectral distance between these two distinct states
	$\rho_1, \rho_2$.	
	This completes the proof.
\end{proof}

\section{Proof of Theorem \ref{thm8}}\label{thm8a}

\begin{proof}
	Define the linear map
	\[ T (a) = i [D, \pi (a)], \qquad a \in \mathcal{A}. \]
	We have $T (\mathbb{I}_n) = 0$, $T (a^{\dagger}) = T (a)^{\dagger}$, and $T
	(a)$ is Hermitian whenever $a$ is Hermitian.
	
	Let $E_j = |j \rangle \langle j|$, $j = 1, \ldots, n$, be the standard diagonal
	matrix units. Since $\sum_{j = 1}^n E_j =\mathbb{I}_n$, we obtain
	\begin{equation}
		\label{eq:T-sum} \sum_{j = 1}^n T (E_j) = 0.
	\end{equation}
	Define the traceless Hermitian matrix
	\[ e_1 = E_1 - \frac{1}{n} \mathbb{I}_n . \]
	Its eigenvalues are $\frac{n - 1}{n}, - \frac{1}{n}$. Therefore, we have
	\[ \|e_1 \|_{op} = \frac{n - 1}{n} . \]
	Since $T (\mathbb{I}_n) = 0$, we have $T (e_1) = T (E_1)$. By the assumption of the theorem,
	\begin{equation}
		\label{eq:T-E1} \|T (E_1)\|_{op} = \|T (e_1)\|_{op} =\|[\mathcal{D}, \pi (e_1)]\|_{op} =
		\|e_1 \|_{op} = \frac{n - 1}{n} = : c_1 > 0.
	\end{equation}
	Since $\mathcal{H}$ is finite-dimensional and $T (E_1)$ is Hermitian, there
	exist a unit vector $|v \rangle  \in \mathcal{H}$ and a sign $s \in \{- 1,
	1\}$ such that
	\begin{equation}
		\label{eq:extremal-vector} T (E_1) |v \rangle  = sc_1 |v \rangle , \qquad
		\langle v  |v \rangle  = 1.
	\end{equation}
	
	For each $j = 2, \ldots, n$, consider the traceless Hermitian matrix $E_1 -
	E_j$, and define
	\[ B_j = s (T (E_1 - E_j)) = s (T (E_1) - T (E_j)) . \]
	The eigenvalues of $E_1 - E_j$ are $1, - 1$, and $0$. Hence, there is
	\begin{equation}\label{eq:T-difference}
		\|E_1 - E_j \|_{op} = 1= \|T (E_1 - E_j) \|_{op}=\|T (E_1) - T (E_j)\|_{op} = \|B_j
		\|_{op}.
	\end{equation}
	Define the real numbers
	\[ b_j = s \langle v | T (E_j) | v \rangle, \qquad j = 2, \ldots, n. \]
	Using \eqref{eq:extremal-vector} and \eqref{eq:T-difference}, we obtain
	\begin{equation}
		c_1 - b_j = s \langle v | T (E_1) - T (E_j) | v \rangle \leqslant \|T (E_1) - T (E_j)\|_{op} = 1. 
	\end{equation}
	Therefore, we have
	\begin{equation}
		\label{eq:b-lower} b_j \geqslant c_1 - 1 = - \frac{1}{n} .
	\end{equation}
	On the other hand, by \eqref{eq:T-sum}, we have
	\begin{equation}
		\sum_{j = 2}^n b_j = s \left\langle v \left| \sum_{j = 2}^n T (E_j)
		\right| v \right\rangle = - s \langle v | T (E_1) | v \rangle = - c_1 = - \frac{n - 1}{n} .  \label{eq:b-sum}
	\end{equation}
	Combining this result with \eqref{eq:b-lower}, we obtain
	\begin{equation}
		\label{eq:b-values} b_j = - \frac{1}{n}, \qquad j = 2, \ldots, n.
	\end{equation}
	Therefore, we have
	\begin{equation}
		s \langle v | T (E_1) - T (E_j) | v \rangle = \langle v | B_j | v \rangle 
		= c_1 - b_j = 1. 
	\end{equation}
	So there is
	\[ \langle v | \mathbb{I} - B_j | v \rangle = 0. \]
	Since $B_j$ is Hermitian and $\|B_j \|_{op} = 1$, we have $\mathbb{I} - B_j \geqslant 0$. 
	For a positive semidefinite operator $K$, the condition $\langle v|K|v \rangle
	= 0$ implies $K|v \rangle = 0$.
	Therefore, we have $(\mathbb{I} - B_j) |v \rangle = 0$, or
	\begin{equation}
		\label{eq:B-eigenvector} s (T (E_1) - T (E_j)) |v \rangle = |v \rangle.
	\end{equation}
	Using \eqref{eq:extremal-vector}, we find
	\[ sT (E_j) |v \rangle = sT (E_1) |v \rangle - |v \rangle = (c_1 - 1) |v \rangle = - \frac{1}{n} |v \rangle. \]
	Hence
	\begin{equation}
		\label{eq:T-Ej} T (E_j) |v \rangle = - \frac{s}{n} |v \rangle = sc_j |v \rangle, \qquad c_j = - \frac{1}{n},\quad j = 2, \ldots,
		n,
	\end{equation}
	
	Define
	\[ Q_j = \pi (E_j), \qquad j = 1, \ldots, n. \]
	Since $\pi$ is a unital $*$-representation, the operators $Q_j$ are orthogonal
	projections satisfying
	\begin{equation}
		\label{eq:projection-relations} Q_j^{\dagger} = Q_j, \qquad Q_j^2 = Q_j,
		\qquad \sum_{j = 1}^n Q_j =\mathbb{I}  .
	\end{equation}
	Using the Leibniz identity for commutators and $Q_j^2 = Q_j$, we have
	\begin{equation}
		{}[D, Q_j] = [D, Q_j^2] = [D, Q_j] Q_j + Q_j  [D, Q_j] . 
	\end{equation}
	Multiplying by $i$ gives
	\begin{equation}
		\label{eq:Leibniz} T (E_j) = T (E_j) Q_j + Q_j T (E_j) .
	\end{equation}
	So there is
	\begin{equation}
		\langle v | T (E_j) | v \rangle = \langle v | T (E_j) Q_j | v \rangle
		+ \langle v | Q_j T (E_j) | v \rangle . 
	\end{equation}
	Since $T (E_j)$ is Hermitian, and $T (E_j) |v \rangle  = sc_j |v \rangle $, we
	obtain
	\begin{equation}
		sc_j = sc_j  \langle v|Q_j| v \rangle + sc_j  \langle v|Q_j| v \rangle
		= 2 sc_j  \langle v| Q_j |v \rangle . 
	\end{equation}
	Since $sc_j \neq 0$, it follows that
	\begin{equation}
		\label{eq:half-projections} \langle v| Q_j |v \rangle = \frac{1}{2}, \qquad j
		= 1, \ldots, n.
	\end{equation}
	Summing \eqref{eq:half-projections} over $j$ and using
	\eqref{eq:projection-relations}, we obtain
	\begin{equation}
		\frac{n}{2} = \sum_{j = 1}^n \langle v| Q_j |v \rangle = \left\langle v\left| \sum_{j = 1}^n Q_j \right| v \right\rangle
		= \langle v| \mathbb{I}| v \rangle = 1. 
	\end{equation}
	Therefore, we have $n = 2$.
	
	When $n=2$, every traceless Hermitian element $e$ satisfies $\lambda_{\max} (e)=-\lambda_{\min} (e)$, and
	\begin{equation}
		e_{\ast} = e - \frac{\lambda_{\max} (e) + \lambda_{\min} (e)}{2}
		\mathbb{I}_n = e .
	\end{equation}
	Therefore, by Theorem \ref{t6}, we also have the result (\ref{ddr}).
\end{proof}


\begin{thebibliography}{100}

	\bibitem{Connes}A. Connes, \textit{Noncommutative geometry}
	(Academic Press, New York, 1994).
	
	\bibitem{Connes1}A. Connes, ``Compact metric spaces, Fredholm modules and
	hyperfiniteness.'' \textit{Ergodic Theory Dynam. Systems}
	\textbf{9} 207-220 (1989).
	
	\bibitem{Bimonte}G. Bimonte, F. Lizzi, G. Sparano, ``Distances on a
	lattice from non-commutative geometry.'' \textit{Phys. Lett. B}
	\textbf{341} 139-146 (1994).
	
	\bibitem{Dai}J. Dai, X.-C. Song, ``Connes' distance of one-dimensional
	lattices: general cases.'' \textit{Commun. Theor. Phys.}
	\textbf{36} 519 (2001).
	
	\bibitem{Cagnache}E. Cagnache, F. D'Andrea, P. Martinetti, J.-C. Wallet,
	``The spectral distance on the Moyal plane.'' \textit{J. Geom.
				Phys.} \textbf{61} 1881-1897 (2011).
	
	\bibitem{Wallet}J. C. Wallet, ``Connes distance by examples: Homothetic
	spectral metric spaces.'' \textit{Rev. Math. Phys.}
	\textbf{24} 1250027 (2012).
	
	\bibitem{Martinetti1}P. Martinetti, F. Mercati, L. Tomassini, ``Minimal
	length in quantum space and integrations of the line element in
	noncommutative geometry.'' \textit{Rev. Math. Phys.}
	\textbf{24} 1250010 (2012).
	
	\bibitem{Martinetti}P. Martinetti, L. Tomassini, ``Noncommutative geometry
	of the Moyal plane: translation isometries, Connes' distance on coherent
	states, Pythagoras equality.'' \textit{Commun. Math. Phys.}
	\textbf{323} 107--141 (2013).
	
	\bibitem{DAndrea}F. D'Andrea, F. Lizzi, J.C. Varilly, ``Metric properties
	of the fuzzy sphere.'' \textit{Lett. Math. Phys.}
	\textbf{103} 183--205 (2013).
	
	\bibitem{Pythagoras}F. D'Andrea, P. Martinetti, ``On Pythagoras theorem
	for products of spectral triples.'' \textit{Lett. Math. Phys.}
	\textbf{103} 469--492 (2013).
	
	\bibitem{Pythagoras1}F. D'Andrea, ``Pythagoras theorem in noncommutative
	geometry'' \textit{Contemporary Mathematics}
	\textbf{676} 175-210 (2016).
	
	\bibitem{Franco}N. Franco, J. C. Wallet, ``Metrics and causality on Moyal
	planes.'' \textit{Contemporary Mathematics}
	\textbf{676} 147-173 (2016).
	
	\bibitem{Scholtz}F. G. Scholtz, B. Chakraborty, ``Spectral triplets,
	statistical mechanics and emergent geometry in non-commutative quantum
	mechanics.'' \textit{J. Phys. A: Math. Theor.}
	\textbf{46} 085204 (2013).
	
	\bibitem{Chaoba}Y. Chaoba Devi, S. Prajapat, A. K. Mukhopadhyay, B.
	Chakraborty, F. G. Scholtz, ``Connes distance function on fuzzy sphere and
	the connection between geometry and statistics.'' \textit{J. Math.
				Phys.} \textbf{56} 041707 (2015).
	
	\bibitem{Revisiting}Y. Chaoba Devi, K. Kumar, B. Chakraborty, F. G.
	Scholtz, ``Revisiting Connes' finite spectral distance on noncommutative
	spaces: Moyal plane and fuzzy sphere.'' \textit{Int. J. Geo. Methods
				Mod. Phys.} \textbf{15} 1850204 (2018).
	
	\bibitem{Kumar}K. Kumar, B. Chakraborty, ``Spectral distances on the
	doubled Moyal plane using Dirac eigenspinors.'' \textit{Phys. Rev.
				D} \textbf{97} 086019 (2018).
	
	\bibitem{Barrett}J. W Barrett, P. Druce, L. Glaser, ``Spectral estimators
	for finite non-commutative geometries.'' \textit{J. Phys. A: Math.
				Theor.} \textbf{52} 275203 (2019).
	
	\bibitem{Chakraborty}A. Chakraborty, B. Chakraborty, ``Spectral distance
	on Lorentzian Moyal plane.'' \textit{Int. J. Geo. Methods Mod.
				Phys.} \textbf{17} 2050089 (2020).
				
	\bibitem{Lin1}B. S. Lin, T. H. Heng, ``Connes distance of $2 D$ harmonic
oscillators in quantum phase space.'' \textit{Chin. Phys. B}
\textbf{30} 110203 (2021).
	
	\bibitem{Lin2}B. S. Lin, T. H. Heng, ``Connes spectral distance and
	nonlocality of generalized noncommutative phase spaces.''
	\textit{Eur. Phys. J. Plus} \textbf{137} 899 (2022).
	
	\bibitem{Clare}P. Clare, C.-K. Li, E. Poon, E. Swartz, ``Noncommutative
	distances on graphs: An explicit approach via Birkhoff-James
	orthogonality.'' \textit{J. Geom. Phys.} \textbf{213}
	105483 (2025).

	\bibitem{Martinettip} P. Martinetti, ``From Monge to Higgs: a survey of distance computations in noncommutative geometry.'' \textit{Contemporary Mathematics.} \textbf{676} 1--46 (2016).


	\bibitem{Nielsen}M. A. Nielsen, I. L. Chuang, \textit{Quantum
				Computation and Quantum Information} (Cambridge University Press,
	Cambridge, 2000).
	
	\bibitem{Iochum}B. Iochum, T. Krajewski, P. Martinetti, ``Distances in
	finite spaces from noncommutative geometry.'' \textit{J. Geom.
				Phys.} \textbf{37} 100--125 (2001).
				
	\bibitem{Martinettif} P. Martinetti, ``Distances en geometrie non commutative.'' [arXiv:math-ph/0112038] (2001).
	

	
	\bibitem{Bassi} J. Bassi, R. Conti, ``On isometries of spectral triples associated to AF-algebras
	and crossed products.'' \textit{J. Noncommut. Geom.} \textbf{18} 547--566 (2024).

	\bibitem{Kumar1} K. Kumar, ``Calibrated Helstrom geometry on the Bloch ball via Connes spectral distance.'' [arXiv:2606.13824] (2026).

   \bibitem{LX} B. S. Lin, Z. H. Xu, J. H. Wang, H. L. Chen,
  ``Connes spectral distances, quantum discord and coherence of qubits.''
  \textit{Phys. Scr.} \textbf{101} 095204 (2026).
  
\end{thebibliography}
\end{document}